\documentclass[submission]{eptcs}

\providecommand{\event}{ACCAT 2012} 
\providecommand{\firstpage}{1}

\usepackage{amssymb}
\usepackage{amsmath}
\usepackage{amsthm}
\usepackage{dsfont}
\usepackage{graphicx}
\usepackage{txfonts}
\usepackage{xcolor}    
  
\DeclareMathAlphabet{\mathpzc}{OT1}{pzc}{m}{it}
     
\newtheorem{definition}{Definition}
\newtheorem{theorem}{Theorem}
\newtheorem{lemma}[theorem]{Lemma}

\newenvironment{example}{%
  \par\noindent
  \textbf{Example}
}{}
\newenvironment{examples}{%
  \par\noindent
  \textbf{Examples}
}{}

\def\SET#1{\lbrace #1 \rbrace}
\def\MSET#1#2{\mlbrace \, #1 \mid #2 \, \mrbrace}
\def\ZSET#1#2{\lbrace \, #1 \mid #2 \, \rbrace}

\def\calA{{\cal A}}

\def\calF{{\cal F}}

\def\calL{{\cal L}}

\def\calP{{\cal P}}

\def\calR{{\cal R}}
\def\calS{{\cal S}}

\def\calV{{\cal V}}

\def\calX{{\cal X}}

\def\endef{\hfill$\bullet$}

\newcommand{\blankline}{\vspace*{0.3\baselineskip}}
\newcommand{\halflineup}{\vspace*{-0.5\baselineskip}}

\newcommand{\myblock}[1]{\par\noindent\textbf{{#1}.}}

\def\bools{\mathbb{B}}
\def\nats{\mathbb{N}}	
\def\reals{\mathbb{R}}

\newcommand{\Set}{\textbf{Set}}
\newcommand{\nil}{\textbf{nil}}
\newcommand{\bfR}{\textbf{R}}
\newcommand{\bfT}{\textbf{T}}

\def\min{\mathrm{\textrm{min}}}

\newcommand{\CCS}{\textsl{CCS}}

\newcommand{\FuTS}{\textsl{Fu\hspace*{-0.5pt}TS}}
\newcommand{\IMC}{\textsl{IMC} \mkern1mu}
\newcommand{\IML}{\textsl{IML} \mkern1mu}
\newcommand{\LTS}{\textsl{LTS}}
\newcommand{\PEPA}{\textsl{PEPA}}
\newcommand{\RTS}{\textsl{RTS}}

\newcommand{\arf}{\textsl{arf} \mkern1mu}

\newcommand{\spt}{\textsl{spt\hspace{1pt}}}

\newcommand{\FALSE}{\texttt{false}}
\newcommand{\TRUE}{\texttt{true}}

\newcommand{\iml}{\mathit{iml}}
\newcommand{\pepa}{\mathit{pepa}}

\newcommand{\finv}{f^{\mkern-1mu -1}}
\newcommand{\sigmainv}{\sigma^{\mkern-1mu -1}}

\newcommand{\amset}[1]{\mathcal{#1}}
\newcommand{\amsetP}{\amset{P}}
\newcommand{\amsetQ}{\amset{Q}}

\newcommand{\Edelay}{\delta}
\newcommand{\Edelaya}{\: \delta_a}
\newcommand{\Edelayb}{\: \delta_b}
\def\Mtrans#1{\stackrel{#1}{\dashrightarrow}}

\newcommand{\OmegaS}{{\Omega_{\mkern2mu \calS}}}
\def\IMPL{\Rightarrow}
\def\IFF{\Leftrightarrow}
\newcommand{\Rclass}[1]{[{#1}]_R}

\newcommand{\Sclass}[1]{[{#1}]_{\mkern2mu \calS}}

\newcommand{\alambda}{(a,\lambda)}
\newcommand{\approxF}{\approx_{\mkern2mu \calF}}
\newcommand{\approxS}{\approx_{\calS}}

\newcommand{\bfzero}{\textbf{0}}
\newcommand{\blambda}{(b,\lambda)}
\newcommand{\bnfeq}{\mathrel{\; \mbox{\emph{\textrm{::=}}} \;}}

\newcommand{\calLi}{\calL_{\mkern2mu i}}
\newcommand{\calLn}{\calL_{\mkern2mu n}}
\newcommand{\calLone}{\calL_{\mkern2mu 1}}

\newcommand{\calRi}{\calR_{\mkern2mu i}}
\newcommand{\calRn}{\calR_{\mkern2mu n}}
\newcommand{\calRone}{\calR_{\mkern2mu 1}}

\newcommand{\calSiml}{\calS_\iml}
\newcommand{\calSpepa}{\calS_\pepa}

\newcommand{\calVLR}{\calV^{\calL \mkern2mu}_{\mkern-4mu \calR \mkern0mu}}
\newcommand{\calVS}{\calV_{\mkern-4mu \calS \mkern0mu}}

\newcommand{\calViml}{\calV_\iml}

\def\chut{\calX}
\def\cho{\, + \, }
\newcommand{\compose}{\mathop{\raisebox{0.5pt}{\scriptsize $\circ$}}}
\def\dfas{:=}

\newcommand{\eqiml}{=_{\iml}}
\newcommand{\eqpepa}{=_{\pepa}}

\newcommand{\fivetuple}[5]{( \mkern1mu {#1}, \mkern1mu {#2}, \mkern1mu {#3}, \mkern1mu {#4}, \mkern1mu {#5} \mkern1mu )}
\newcommand{\fmorph}[3]{[\![ {#3} ]\!]^{#1}_{#2}}

\newcommand{\fmorphS}[2]{\fmorph{\mkern2mu \calS}{#1}{#2}}

\newcommand{\fsfn}[2]{\mkern1mu \mathrm{\mathcal{F} \mkern-2.5mu \mathcal{S}}( \mkern1mu #1,#2 \mkern2mu )}

\def\fsum{\mathopen{\oplus \mkern1mu}}

\newcommand{\la}{\mathop{\langle \,}}

\newcommand{\mlbrace}{\mathopen{\lbrace \mkern-2.25mu | \mkern2.25mu}}
\newcommand{\mrbrace}{\mathclose{\mkern2mu | \! \rbrace}}

\newcommand{\mkell}{{\mkern1mu \ell}}
\def\mtrans#1{\stackrel{#1}{\rightarrowtail}}

\newcommand{\myell}{{\ell \mkern2mu}}
\newcommand{\myf}{{f \mkern2mu}}
\newcommand{\myi}{{\mkern1mu i \mkern1mu}}
\newcommand{\myin}{\, \in \,}
\newcommand{\myop}{\mathbin{\mkern1mu | \mkern1mu}}

\newcommand{\mytheta}{\theta \mkern1mu}

\def\nnreals{\reals_{\geqslant 0}}

\newcommand{\prlA}{\mathbin{\mkern2mu \parallel_{\mkern1mu A} \mkern0mu}}

\def\pfx#1#2{#1.#2}
\def\prc#1{{\calP}_{#1}}

\newcommand{\prcIML}{\prc{\IML}}
\newcommand{\prcPEPA}{\prc{\PEPA}}

\newcommand{\ra}{\mathclose{\, \rangle}}
\def\poreals{\reals_{> 0}}

\newcommand{\simS}{\sim_{\calS}}

\newcommand{\simiml}{\sim_\iml}

\newcommand{\simpepa}{\sim_{\pepa}}

\def\sosrule#1#2{\frac{\,\,\,\,#1\,\,\,\,}{\,\,\,\,#2\,\,\,\,}}
\def\sosrule#1#2{\begin{array}{c} {#1}\rule{0pt}{15pt} \\ \hline \rule{0pt}{15pt}{#2} \end{array}}
\def\sosrn#1#2{\mbox{{\normalsize (#1$_{\mbox{{\scriptsize #2}}}$)\ }}}

\newcommand{\thetai}{\theta_\myi}
\newcommand{\thetapepa}{\theta_\pepa}
\newcommand{\threetuple}[3]{( \mkern1mu {#1}, \mkern1mu {#2}, \mkern1mu {#3} \mkern1mu )}
\newcommand{\trans}[1]{\stackrel{#1}{\rightarrow}}
\newcommand{\transalambda}{\trans{a,\lambda}}
\newcommand{\tssum}{\textstyle{\sum \,}}
\newcommand{\twotuple}[2]{( \mkern1mu {#1}, \mkern1mu {#2}\mkern1mu )}

\newcommand{\zerofn}{\underline{\mathrm{0}}}
\renewcommand{\zerofn}{[ \mkern1mu ]}
\def\zerof{\zerofn}

\title{Bisimulation of Labeled State-to-Function Transition Systems 
  \\ of Stochastic Process Languages}
\author{%
  D. Latella
  \&
  M. Massink
  \institute{CNR -- Istituto di Scienza e Tecnologie dell'Informazione `A. Faedo'}
  \and
  E.P. de Vink\,\thanks{%
    Corresponding author, email~\url{evink@win.tue.nl}.}
  \institute{Technische Universiteit Eindhoven and Centrum Wiskunde Informatica}
}

\def\titlerunning{Bisimulation of $\FuTS$}
\def\authorrunning{Latella, Massink \& De Vink}

\begin{document}
\maketitle

\begin{abstract}
  \textbf{Abstract}
  Labeled state-to-function transition systems, $\FuTS$ for short, admit multiple transition schemes from states to functions of finite support over general semirings.
  As such they constitute a convenient modeling instrument to deal with stochastic process languages.
  In this paper, the notion of bisimulation induced by a $\FuTS$ is addressed from a coalgebraic point of view. 
  A correspondence result is proven stating that $\FuTS$-bisimulation coincides with the behavioral equivalence of the associated functor.
  As generic examples, the concrete existing equivalences for the core of the stochastic process algebras $\PEPA$ and~$\IML$ are related to the bisimulation of specific~$\FuTS$, providing via the correspondence result coalgebraic justification of the equivalences of these calculi.
\end{abstract}


\section{Introduction}

  Process description languages equipped with formal operational semantics are successful formalisms for modeling concurrent systems and analyzing their behavior.  
  Typically, the operational semantics is defined by means of a labeled transition system following the SOS approach.
  The states of the transition systems are just process terms, while the labels of the transitions between states represent the possible actions and interactions.  
  Process description languages often come equipped with process equivalences, so that system models can be compared according to specific behavioral relations.  

  In the last couple of decades, process languages have been enriched with quantitative information. 
  Among these quantitative extensions, those allowing a stochastic representation of time, usually referred to as stochastic process algebras, have received particular attention.   
  The main aim has been the integration of qualitative descriptions and quantitative analysis in a single mathematical framework by building on the combination of labeled transition systems and continuous-time Markov chains. 
  The latter being one of the most successful approaches to modeling and analyzing the performance of computer systems and networks.   
  An overview on stochastic process algebras, equivalences and related analysis techniques can be found in~\cite{HHK02,BHHKS04:voss,Ber07:sfm}, for example.  
  A common feature of many stochastic process algebras is that actions are enriched with the rates of exponentially distributed random variables that characterize their duration.
  Although exploiting the same class of distributions, the models and the techniques underlying the definition of the calculi turn out to be significantly different in many respects. 
  A prominent difference concerns the modeling of the race condition by means of the choice operator, and its relationship to the issue of transition multiplicity. 
  In the quantitative setting, multiplicities can make a crucial distinction between processes that are qualitatively equivalent.
  Several significantly different approaches have been proposed for handling transition multiplicity. 
  The proposals range from multi-relations~\cite{Hil96:phd,Her02:springer}, to proved transition systems~\cite{Pri95:cj}, to $\LTS$ with numbered transitions~\cite{HHK02}, to unique rate names~\cite{DLM05}, just to mention a few.

\pagebreak[3]

In~\cite{De+09}, Latella, Massink et al.\ have proposed a variant of $\LTS$, called Rate Transition Systems ($\RTS$). 
  In~$\LTS$, a transition is a triple $(P,\alpha,P' \mkern1mu )$ where $P$ and~$\alpha$ are the source state and the label of the transition, respectively, while $P'$ is the target state reached from~$P$ via the transition. 
  In~$\RTS$, a transition is a triple of the form $(P,\alpha,\amsetP \mkern2mu )$. 
  The first and second component are the source state and the label of the transition, as in~$\LTS$, while the third component~$\amsetP$ is a continuation function which associates a non-negative real value to each state~$P'$.  
  A non-zero value for the state~$P'$ represents the rate of the exponential distribution characterizing the time for the execution of the action represented by~$\alpha$, necessary to reach $P'$ from~$P$ via the transition. 
  If~$\amsetP$~maps $P'$ to~$0$, then state~$P'$ is not reachable from~$P$ via the transition. 
  The use of continuation functions provides a clean and simple solution to the transition multiplicity problem and make~$\RTS$ particularly suited for stochastic process algebra semantics.
  In order to provide a uniform account of the many stochastic process algebras proposed in the literature, in previous joint work of the first two authors~\cite{De+11} Labelled State-to-Function Transition Systems ($\FuTS$) have been introduced as a natural generalization of~$\RTS$. 
  In~$\FuTS$ the co-domain of the continuation functions are arbitrary semirings, rather than just the non-negative reals.  
  This provides increased flexibility while preserving basic properties of primitive operations like sum and multiplication.      

  In this paper we present a coalgebraic treatment of~$\FuTS$ that allow multiple state-to-function transition relations involving arbitrary semirings.
  Given label sets~$\calLi$ and semirings~$\calRi$, a~$\FuTS$ takes the general format $\calS = ( \, S ,\, \la \mkern-2mu {\mtrans{}_i} \ra^n_{i = 1} \, )$ with transition relations ${\mtrans{}_i} \, \subseteq \, S \times \calLi \times \fsfn{\mkern2mu S}{\calRi}$.
  Here, $\fsfn{\mkern2mu S}{\calRi}$ are the sets of functions from~$S$ to~$\calRi$ of finite support, a subcollection of functions also occurring in other work combining coalgebra and quantitative modeling.
  We will associate to~$\calS$ the product of the functors $\fsfn{ {\cdot}}{\calRi}^{\mkern1mu \calLi}$.
  For this to work, we need the transition relations~$\mtrans{}_i$ to be total and deterministic for the coalgebraic modeling as a function.
  Maybe surprisingly, this isn't a severe restriction at all in the presence of continuation functions: 
  the zero-continuation $\lambda s' . \mkern1mu 0$ expresses that no $\LTS$-transition exists from the state~$s$ to any state~$s'$; if $s$~allows a transition to some state~$s_1$ as well as a state~$s_2$, the continuation function will simply yield a non-zero value for $s_1$ and for~$s_2$.

  The notion of $\calS$-bisimulation that arises from a~$\FuTS$~$\calS$ is reinterpreted coalgebraically as the behavioral equivalence of a functor that is induced by~$\calS$, along the lines sketched above.
  Behavioral equivalence rather than coalgebraic bisimulation is targeted, since, dependent on the semirings involved, weak pullbacks may not be preserved and the construction of a mediating morphism for a coalgebraic bisimulation from a concrete one may fail for degenerate denominators.
  However, following a familiar argument, we show that the functor associated with a~$\FuTS$ does possess a final coalgebra and therefore has an associated notion of behavioral equivalence indeed.
  It is noted, in the presence of a final coalgebra for $\FuTS$ a more general definition of behavioral equivalence based on cospans coincides~\cite{Kur00:phd}.
  A correspondence result is proven in this paper that shows that the concrete bisimulation of a~$\FuTS$, coincides with behavioral equivalence of its functor.
  Pivotal for its proof is the absence of multiplicities in the $\FuTS$ treatment of quantities.
  
  Using the bridge established by the correspondence result, we continue by showing for two well-known stochastic process algebras, viz.\ Hillston's $\PEPA$~\cite{Hil96:phd} and Hermanns's $\IML$~\cite{Her02:springer}, that the respective standard notion of strong equivalence and strong bisimulation coincides with behavioral equivalence of the associated~$\FuTS$.
  This constitutes the main contribution of the paper.
  $\PEPA$ stands out as one of the prominent Markovian process algebras, while $\IML$ specifically provides separate prefix constructions for actions and for delays.
  The equivalences of~$\PEPA$ and of~$\IML$ are compared with the bisimulations of the respective~$\FuTS$ as given by an alternative operational semantics involving the state-to-function scheme.
  In passing, the multiplicities have to be dealt with. 
  Appropriate lemmas are provided relating the relation-based cumulative treatment with $\FuTS$ to the multirelation-based explicit treatment of $\PEPA$ and~$\IML$.

  Related work on coalgebra includes \cite{DeR99,KS08:fossacs,Sok11:tcs}, papers that also cover measures and congruence formats, a topic not touched upon here.
  For the discrete parts, regarding the correspondence of bisimulations, our work aligns with the approach of the papers mentioned.
  In this paper the bialgebraic perspective of SOS and bisimulation~\cite{TP97:lics} is left implicit.
  An interesting direction of research combining coalgebra and quantities studies various types of weighted automata, including linear weighted automata, and associated notions of bisimulation and languages, as well as algorithms for these notions~\cite{Bor09:concur,SBBR10:ic,BBBRS11}.
  In particular, building on a result on bounded functors~\cite{GS01:au}, it is shown in~\cite{BBBRS11} for a functor involving functions of finite support over a field that the final coalgebra exists. 
  Below, we have followed the scheme of~\cite{BBBRS11} to obtain such a result for a functor induced by a~$\FuTS$.
  The notions of equivalence addressed in this paper, as often in coalgebraic treatments of process relations, are all strong bisimilarities.

  The present paper is organized as follows: 
  Section~\ref{sec-preliminaries} briefly discusses some material on semirings and coalgebras. 
  Labeled state-to-function transition systems and $\FuTS$ as well as the associated notion of bisimulation are provided in Section~\ref{sec-futs}.
  The coalgebraic counterparts of $\FuTS$ and $\FuTS$-bisimulation are defined in Section~\ref{sec-coalgebra}, where we also establish the correspondence with behavioral equivalence of the final coalgebra.
  In~Section~\ref{sec-pepa} the standard equivalence of~$\PEPA$ is identified with the bisimulation of a $\FuTS$ and, hence, with behavioral equivalence.
  In~Section~\ref{sec-iml} the same is done for the language of~$\IMC$ where actions and delays are present on equal footing.
  Section~\ref{sec-conclusions} wraps up and discusses directions of future research.
  An appendix provides the proofs of a number of lemmas.


\section{Preliminaries}
\label{sec-preliminaries}

\noindent
  A tuple $\calR = \fivetuple{R}{+}{0}{\ast}{1}$ is called a semiring, if $\threetuple{R}{+}{0}$ is a commutative monoid with neutral element~$0$, $\threetuple{R}{\ast}{1}$ is a monoid with neutral element~$1$, $\ast$~distributes over~$+$, and $0 \ast r = r \ast 0 = 0$ for all~$r \in R$.
  As examples of a semiring we will use are the booleans
$\bools = \SET{ \, \FALSE ,\, \TRUE \, }$ with disjunction as sum and conjunction as multiplication, and the non-negative reals~$\nnreals$ with the standard operations.
  We will consider, for a semiring~$\calR$ and a function $\varphi : X \to \calR$, countable sums $\tssum_{x \myin X'} \; \varphi(x)$ in~$\calR$, for~$X' \subseteq X$. For such a sum to exist we require~$\varphi$ to be of finite support, i.e.\ the support set $\spt(\varphi) = \ZSET{ x \in X }{ \varphi(x) \neq 0 }$ is finite.
  Here, $0$~is the neutral element of~$\calR$ with respect to~$+$.
  
  We use the notation $\fsfn{X}{\calR}$ for the collection of all functions of finite support from the set~$X$ to the semiring~$\calR$.
  A construct $[ \, x_1 \mapsto r_1 ,\, \ldots ,\, x_n \mapsto r_n \,]$, with $x_i \in X$, $i = 1 \ldots n$ all distinct, $r_i \in \calR$, $i = 1 \ldots n$, denotes the mapping that assigns $r_i$ to~$x_i$, $i = 1 \ldots n$, and assigns $0$ to all $x \in X$ different from all~$x_i$.
  In particular~$\zerof \mkern1mu$, or more precisely~$\zerof_{\mkern1mu \calR}$, is the constant function~$x \mapsto 0$ and $\chut_{x} = [ \,  x \mapsto 1 \, ]$ is the characteristic function on~$\calR$ for $x \in X$.
  For $\varphi \in \fsfn{X}{\calR}$, we write $\fsum \varphi$ for the value $\tssum_{x \in X} \; {\varphi(x)}$ in~$\calR$.
  For $\varphi, \psi \in \fsfn{X}{\calR}$, the function $\varphi \cho \psi$ is the pointwise sum of $\varphi$ and~$\psi$, i.e.\ $(\varphi \cho \psi)(x) = \varphi(x) \cho \psi(x) \in \calR$.
  Clearly, $\varphi \cho \psi$ is of finite support as $\varphi$ and~$\psi$ are.
  Given an injective operation~${\myop} \colon X \times X \to X$, we define $\varphi \myop \psi : X \to \calR$, by $( \varphi \myop \psi )( x ) = \varphi(x_1) \ast \psi(x_2)$ if $x = x_1 \myop x_2$ for some $x_1, x_2 \in X$, and $( \varphi \myop \psi )( x ) = 0$ otherwise.
  Again, $\varphi \myop \psi$ is of finite support as $\varphi$ and~$\psi$ are.
  This is used in the setting of syntactic processes~$P$ that may have the form $P_1 \prlA P_2$ for two processes $P_1$ and~$P_2$ and a syntactic operator~$\prlA$.

\blankline

\begin{lemma}
\label{lm-props-fsum}
  Let~$X$ be a set, $\calR$ a semiring, and $\myop$ an injective binary operation on~$X$.
  For $\varphi, \psi \in \fsfn{X}{\calR}$ it holds that $\fsum (\varphi \cho \psi ) = \fsum \varphi \cho \fsum \psi$ and $\fsum ( \, \varphi \myop \psi \, ) = ( \fsum \varphi ) \ast ( \fsum \psi )$.
  \qed
\end{lemma}

\blankline

\noindent
  We recall some basic definitions from coalgebra.
  See e.g.~\cite{Rut00:tcs} for more details.
  For a functor $\calF : \Set \to \Set$ on the category $\Set$ of sets and functions, a coalgebra of~$\calF$ is a set~$X$ together with a mapping $\alpha : X \to \calF(X)$.
  A homomorphism between two $\calF$-coalgebras $(X,\alpha)$ and~$(Y,\beta)$ is a function $f : X \to Y$ such that $\calF(f) \compose \alpha = \beta \compose f$.
  An $\calF$-coalgebra $(\Omega,\omega)$ is called final, if there exists, for every $\calF$-coalgebra $(X,\alpha)$, a unique homomorphism $\fmorph{\calF}{X}{\cdot} : (X,\alpha) \to (\Omega,\omega)$.
  Two elements~$x_1, x_2$ of a coalgebra~$(X,\alpha)$ are called behavioral equivalent with respect to~$\calF$ if $\fmorph{\calF}{X}{x_1} = \fmorph{\calF}{X}{x_2}$, notation $x_1 \approxF x_2$.

  Using a characterization of~\cite{GS02:mscs}, a functor $\calF$ on~$\Set$ is bounded, if there exist sets $A$ and~$B$ and a surjective natural transformation $\eta : A \times ( {\cdot} )^B \mathrel{\Rightarrow} \calF$. 
  Here, $A \times ( {\cdot} )$ is the functor that maps a set~$X$ to the Cartesian product $A \times X$ and maps a function $f : X \to Y$ to the mapping $A \times f : A \times X \to A \times Y$ with $(A \times f)(a,x) = (a,f(x))$, while $( {\cdot} )^B$ denotes the functor that maps a set~$X$ to the function space $X^B$ of all functions from $B$ to~$X$ and that maps a function $f : X \to Y$ to the mapping $f^B : X^B \to Y^B$ with $f^B( \varphi )( b ) = f( \varphi( b ) )$.
  For bounded functors we have the following result, see~\cite{GS01:au} for a proof.

\blankline

\begin{theorem}
\label{th-existence-final-coalgebra}
  If a functor $\calF : \Set \to \Set$ is bounded, then its final coalgebra exists.
  \qed
\end{theorem}

\blankline

\noindent
  A number of proofs of results on process languages~$\prc{}$ in this paper relies on so-called guarded recursion~\cite{BV96:mit}.
  Typically, constants~$X$ are a syntactical ingredient in these languages.
  As usual, if $X \dfas P$, i.e.\ the constant~$X$ is declared to have the process~$P$ as its body, we require $P$ to be prefix-guarded.
  Thus, any occurrence of a constant in the body~$P$ is in the scope of a prefix-construct of the language.
  Guarded recursion assumes the existence of a function $c : \prc{} \to \nats$ such that $c(P_1 \mathbin{\raisebox{0.5pt}{\fontsize{9}{0}{$\bullet$}}} P_2) > \textrm{max} \SET{ \, c(P_1) ,\, c(P_2) \,}$ for all syntactic operations~\raisebox{0.5pt}{\fontsize{9}{0}{$\bullet$}} of~$\prc{}$, and moreover $c(X) > c(P)$ if $X \dfas P$.


\section{Labeled State-to-Function Transition Systems}
\label{sec-futs}

The definition of a labeled state-to-function transition system, $\FuTS$ for short, involves a set of states~$S$ and one or more relations of states and functions from states into a semiring.
  For sums over arbitrary subsets of states to exist, the functions are assumed to be of finite support.

\begin{definition}
  \label{df-ltfs}
  A $\FuTS$~$\calS$, in full a labeled state-to-function transition system, over a number of label sets~$\calLi$ and semirings~$\calRi$, $i = 1 \ldots n$, is a tuple $\calS = ( \, S ,\, \la \mkern-2mu {\mtrans{}_i} \ra^n_{i = 1} \, )$ such that ${\mtrans{}_i} \, \subseteq \, S \times \calLi \times \fsfn{\mkern2mu S}{\calRi}$, for $i = 1 \ldots n$.
\endef
\end{definition}

\noindent
  As usual, we write $s \mtrans{\ell}_i v$ for $(s,\ell,v) \in {\mtrans{}_i}$.
  For a $\FuTS$~$\calS = ( \, S ,\, \la \mkern-2mu {\mtrans{}_i} \ra^n_{i = 1} \, )$ the set~$S$ is called the set of states. 
  We refer to each~$\mtrans{}_i$ as a state-to-function transition relation of~$\calS$ or just as a transition relation of it.
  If for~$\calS$ we have that $n = 1$, i.e.\ there is only one state-to-function transition relation~$\mtrans{}$, then $\calS$ is called simple.
  A $\FuTS$~$\calS$ is called total and deterministic if for each transition relation ${\mtrans{}_i} \subseteq S \times \calLi \times \fsfn{S}{\calRi}$ involved and for all $s \in S$, $\ell \in \calLi$, we have $s \mtrans{\ell}_i v$ for exactly one $v \in \fsfn{S}{\calRi}$.
  In such a situation, the zero-function $\zerof_{\calRi}$ plays a special role.
  A state-to-function transition $s \mtrans{\ell}_i \zerof_{\calRi}$ reflects the absence of a non-trivial transition $s \mtrans{\ell}_i v$ for $v \neq \zerof_{\calRi}$.
  In the context of $\LTS$ one says that $s$~has no $\ell$-transition.
  For the remainder of the paper, all $\FuTS$ we consider are assumed to be total and deterministic.\footnote{Definition~\ref{df-ltfs} slightly differs in formulation from the one in~\cite{De+11}.}

\blankline

\begin{examples}
  For the modeling of $\CCS$ processes, we choose a set of actions~$\calA$ as label set and the booleans~$\bools$ as semiring.  
  Consider the two $\CCS$ processes $P = a.b.\bfzero \cho a.c.\bfzero$ and $Q = a.(b.\bfzero \cho c.\bfzero)$, their representation as a $\FuTS$ is depicted in Figure~\ref{fig-examples}.
  For process~$P$ we have $P \mtrans{a} [ b.\bfzero \mapsto \TRUE, \, c.\bfzero \mapsto \TRUE ]$, while the process~$Q$ we have $Q \mtrans{a} [ b.\bfzero \cho c.\bfzero \mapsto \TRUE ]$.
  So, $\FuTS$ are able to represent branching.

  As another example of a simple $\FuTS$, Figure~\ref{fig-examples} displays a $\FuTS$ over the action set~$\calA$ and the semiring~$\nnreals$ of the non-negative real numbers. 
  The functions $v_0$ to~$v_4$ used in the example have the property that $\fsum v_\myi(s) \: = \: 1$, for $i = 0 \ldots 4$.
  Usually, such a $\FuTS$ over~$\nnreals$ is called a (reactive) probabilistic transition system.
  
  In Section~\ref{sec-iml} we will provide semantics for the process language~$\IML$ of interactive Markov \linebreak chains~\cite{Her02:springer,HK10:fmco} using $\FuTS$.
  Unlike many other stochastic process algebras, a single $\IML$ process can in general both perform action-based transitions and time-delays governed by exponential distributions.
\end{examples}

\begin{figure}
  \centering
  \scalebox{0.90}{%
  \includegraphics[scale=0.60]{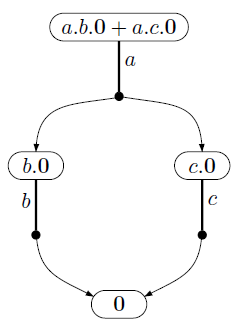}
  \qquad
  \includegraphics[scale=0.60]{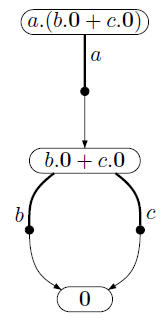}
  \qquad \qquad
  \includegraphics[scale=0.55]{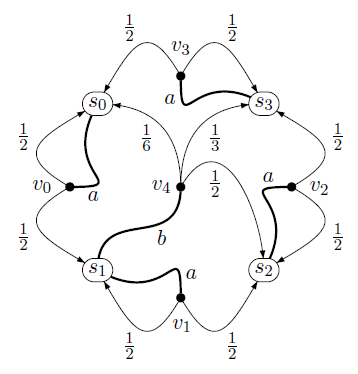}
  } 
  \halflineup
  \halflineup
  \caption{$\FuTS$ for two $\CCS$ processes and a probabilistic process.}
  \label{fig-examples}
\end{figure}


\blankline

\noindent
  It will be notationally convenient to consider a (total and deterministic) $\FuTS$ as a tuple $( \, S ,\, \la \mkern-2mu {\theta_i} \ra^n_{i = 1} \, )$ with transition functions $\theta_i : S \to \calLi \to \fsfn{\mkern2mu S}{\calRi}$, $i = 1 \ldots n$, rather than using the form $( \, S ,\, \la \mkern-2mu {\mtrans{}_i} \ra^n_{i = 1} \, )$ that occurs more frequent for concrete examples in the literature. 
  Alternatively, using disjoint unions, one could see a $\FuTS$ represented by a function $\theta' : S \to \bigoplus_{i=1}^n \: \calLi \to \bigoplus_{i=1}^n \: \fsfn{S}{\calRi}$ satisfying the additional property that $\theta'(s)( \myell ) \in \fsfn{S}{\calRi}$ if $\ell \in \calLi$.
  As this fits less smoothly with the category-theoretical approach of Section~\ref{sec-coalgebra}, we stick to the former format.
  Note, an interpretation of a $\FuTS$ as a function $S \to \bigoplus_{i=1}^n \: \bigl( \, \calLi \to \fsfn{S}{\calRi} \, \bigr)$ does not suit our purposes as the $\IML$ example above illustrates.
  
  We will use the notation with transition functions $\thetai : S \to \calLi \to \fsfn{\mkern2mu S}{\calRi}$ to introduce the notion of bisimilarity for a $\FuTS$.

\blankline

\begin{definition}
  \label{df-ltfs-bisim}
  Let $\calS = ( \, S ,\, \la \mkern-2mu {\theta_i} \ra^n_{i = 1} \, )$ be a $\FuTS$ over the label sets $\calLi$ and semirings~$\calRi$, $i = 1 \ldots n$.
  An equivalence relation $R \subseteq S \times S$ is called an $\calS$-bisimulation if $R(s_1,s_2)$ implies
\begin{equation}
  \tssum_{t' \in \Rclass{t}} \; \theta_i \mkern2mu (s_1)(\ell \mkern2mu )(t' \mkern1mu )
  =
  \tssum_{t' \in \Rclass{t}} \; \theta_i \mkern2mu (s_2)(\ell \mkern2mu )(t' \mkern1mu )
  \label{eq-ltfs-bisim}
\end{equation}%
for all $t \in S$, $i = 1 \ldots n$ and $\ell \in \calLi$.
  Two elements $s_1, s_2 \in S$ are called $\calS$-bisimilar if $R(s_1,s_2)$ for some $\calS$-bisimulation~$R$ for~$\calS$.
  Notation $s_1 \simS s_2$.
\endef
\end{definition}

\noindent
  We use the notation $\Rclass{t}$ to denote the equivalence class of~$t \in S$ with respect to~$R$.
  Note that sums in equation~(\ref{eq-ltfs-bisim}) exist since the functions $\theta_i \mkern2mu (s_j)(\ell \mkern2mu ) \in \fsfn{\mkern2mu S}{\calRi}$, $i = 1 \ldots n$, $j = 1,2$, are of finite support. 
  Hence, $\theta_i \mkern2mu (s_j)(\ell \mkern2mu )(t') = 0 \in \calRi$ for all but finitely many~$t' \in \Rclass{t} \subseteq S$.
  
\blankline

\noindent
  For the combined $\FuTS$ of the two $\CCS$-processes of Figure~\ref{fig-examples}, the obvious equivalence relation relating $b.\bfzero$ and~$b.\bfzero \cho c.\bfzero$ is not a $\FuTS$-bisimulation.
  Although $\tssum_{t' \in \Rclass{\bfzero}} \; \mytheta(b.\bfzero)(b)(t') = \theta(b.\bfzero)(b)(\bfzero) = \TRUE$ and $\tssum_{t' \in \Rclass{\bfzero}} \; \theta(b.\bfzero \cho c.\bfzero)(b)(t') = \theta(b.\bfzero \cho c.\bfzero)(b)(\bfzero) = \TRUE$, we have $\tssum_{t' \in \Rclass{\bfzero}} \; \theta(b.\bfzero)(c)(t') = \FALSE$, while $\tssum_{t' \in \Rclass{\bfzero}} \; \theta(b.\bfzero \cho c.\bfzero)(c)(t') = \TRUE$, taking sums, i.e.\ disjunctions, in~$\bools$.
  


\section{$\FuTS$ coalgebraically}
\label{sec-coalgebra}

  In this section we will cast $\FuTS$ in the framework of coalgebras and prove a correspondence result of $\FuTS$-bisimulation and behavioral equivalence for a suitable functor on~$\Set$.

\blankline

\begin{definition}
  \label{df-v-functor}
  Let $\calL$ be a set of labels and let~$\calR$ be a semiring. 
  The functor $\calVLR : \Set \to \Set$ assigns to a set~$X$ the function space $\fsfn{X}{\calR}^\calL$ of all functions $\varphi : \calL \to \fsfn{X}{\calR}$ and assigns to a function $f : X \to Y$ the mapping $\calVLR(\myf) : \fsfn{X}{\calR}^\calL \to \fsfn{Y}{\calR}^\calL$ where 
\begin{displaymath}
  \calVLR(\myf)(\varphi)(\myell)(y)
  =
  \tssum_{x' \in \finv(y)} \  \varphi(\myell)(x')
\end{displaymath}
for all $\varphi \in \fsfn{X}{\calR}^\calL$, $\ell \in \calL$ and~$y \in Y$.
\endef
\end{definition}

\blankline

\noindent
  Again we rely on~$\varphi(\myell) \in \fsfn{X}{\calR}$ having a finite support for the sum to exist and for $\calVLR$ being well-defined.
  In fact, we have $\spt( \, \calVLR( \myf )(\varphi)(\myell) \, ) = \ZSET{ f(x) }{ x \in \spt(\varphi)(\myell) }$.
  
  As we aim to compare our notion of bisimulation for~$\FuTS$ with behavioral equivalence for the functor~$\calVLR$, given a set of labels~$\calL$ and a semiring~$\calR$, we need to check that $\calVLR$ possesses a final coalgebra. 
  We follow the approach of~\cite{BBBRS11}. 
  
\blankline

\begin{lemma}
\label{lm-v-is-bounded}
  Let~$\calL$ be a set of labels, $\calR$~a semiring.
  Then the functor~$\calVLR$ on~$\Set$ is bounded.
  \qed
\end{lemma}

\blankline

\noindent
  Working with total and deterministic $\FuTS$, we can interpret a $\FuTS$ $\calS = ( \, S ,\, \la \mkern-2mu {\theta_i} \ra^n_{i = 1} \, )$ over the label sets $\calLi$ and semirings~$\calRi$, $i = 1 \ldots n$ as a product $\theta_1 \times \cdots \times \theta_n : S \to \prod_{i=1}^n \: ( \, \calLi \to \fsfn{S}{\calRi} \, )$ of functions $\thetai : S \to \calLi \to \fsfn{S}{\calRi}$.
  To push this idea a bit further, we want to consider the $\FuTS$~$\calS = ( \, \calS ,\, \la \mkern-2mu {\thetai} \ra^n_{i = 1} \,)$ as a coalgebra of a suitable product functor on~$\Set$. 
  
\blankline
  
\begin{definition}
\label{df-futs-functor}
  Let $\calS = ( \, S ,\, \la \mkern-2mu {\theta_i} \ra^n_{i = 1} \, )$ be a $\FuTS$ over the label sets $\calLi$ and semirings~$\calRi$, $i = 1 .. n$. The functor~$\calVS$ on~$\Set$ is defined by $\calVS = \prod_{i=1}^n \: \calV_{\calRi}^{\mkern1mu \calLi} = \prod_{i=1}^n \: \fsfn{ {\, \cdot \,} }{\calRi}^{\calLi}$.
\end{definition}

\blankline

\noindent
  The point is, under conditions that are generally met, coalgebras come equipped with a natural notion of behavioral equivalence that can act as a reference for strong equivalences, in particular of  bisimulation for~$\FuTS$.
  Below, see Theorem~\ref{th-correspondence}, we prove that $\calS$-bisimilarity as given by Definition~\ref{df-ltfs-bisim} coincides with behavioral equivalence for the functor~$\calVS$ as given by Definition~\ref{df-futs-functor}, providing justification for the notion of equivalence defined on~$\FuTS$.
   
  For the notion of behavioral equivalence for the functor~$\calVS$ obtained from~$\calS$ to be defined, we establish that it possesses a final coalgebra.
  
\begin{theorem}
\label{th-v-has-final-coalgebra}
  The functor $\calVS$ has a final coalgebra.
\end{theorem} 
\begin{proof}
 By Lemma~\ref{lm-v-is-bounded} we have that each factor $\calV_{\calRi}^{\mkern1mu \calLi}$ of~$\calVS$ is bounded, and hence possesses a final coalgebra~$\Omega_{\calV_{\calRi}^{\mkern1mu \calLi}}$ by Theorem~\ref{th-existence-final-coalgebra}.
 It follows that also $\calVS$ has a final coalgebra~$\OmegaS$. 
 Writing $\fmorphS{X}{\cdot}$ for the final morphism of a $\calVS$-coalgebra~$X$ into~$\OmegaS$, we have 
\begin{displaymath}
  \OmegaS = \Omega_{\calV_{\calRone}^{\mkern1mu \calLone}} \times \cdots \times \Omega_{\calV_{\calRn}^{\mkern1mu \calLn}}
  \quad \text{and} \quad
  \fmorphS{X}{\cdot} = \fmorph{\calV_{\calRone}^{\mkern1mu \calLone}}{X}{\cdot} \times \cdots \times \: \fmorph{\calV_{\calRn}^{\mkern1mu \calLn}}{X}{\cdot}
\end{displaymath}
as can be straightforwardly shown.
\end{proof}

\noindent
  Since the functor~$\calVS$ of a $\FuTS$ has a final coalgebra, we can speak of the behavioral equivalence~$\approxS$ induced by~$\calVS$.
  Next we establish, for a $\FuTS$~$\calS$, the correspondence of $\calS$-bisimulation~$\simS$ as given by Definition~\ref{df-ltfs-bisim} and behavioral equivalence~$\approxS$.
  
\begin{theorem}
  \label{th-correspondence}
  Let $\calS = ( \, S ,\, \la \mkern-2mu {\theta_i} \ra^{\, n}_{i = 1} \, )$ be a $\FuTS$ over the label sets $\calLi$ and semirings~$\calRi \mkern1mu$, $i = 1 \ldots n$.
  Then $s_1 \simS s_2 \IFF s_1 \approxS s_2$, for all $s_1, s_2 \in S$.
\end{theorem}
\begin{proof}
  Let $s_1, s_2 \in S$.
  We first prove $s_1 \simS s_2 \IMPL s_1 \approxS s_2$.
  So, assume $s_1 \simS s_2$.
  Let $R \subseteq S \times S$ be an $\calS$-bisimulation with $R(s_1,s_2)$.
  Put $\theta = \theta_{\mkern0mu 1} \times \cdots \times \theta_{\mkern1mu n}$.
  Note $\twotuple{S}{\mytheta}$ is a $\calVS$-coalgebra.
  We turn the collection of equivalence classes~$S/R$ into a $\calVS$-coalgebra $(S/R,\theta_R)$ by putting 
\begin{displaymath}
  \theta^\myi_{R}( \, \Rclass{s} \, )(\myell)( \, \Rclass{t} \, ) = 
  \textstyle{\sum}_{t' \in \Rclass{t}} \; \theta_\myi(s)(\myell)(t')
  \quad \text{and} \quad
  \theta_R = \theta^{\mkern2mu 1}_{R} \times \cdots \times \theta^{\mkern2mu n}_{R}
\end{displaymath}%
for $s, t \in S$, $\ell \in \calLi$, $i = 1 \ldots n$.
  This is well-defined since $R$ is an $\calS$-bisimulation: if $R(s,s')$ then we have $\textstyle{\sum}_{t' \in \Rclass{t}} \; \theta_\myi(s)(\myell)(t') = \textstyle{\sum}_{t' \in \Rclass{t}} \; \theta_\myi(s')(\myell)(t')$.
  The canonical mapping $\varepsilon_R : S \to S/R$ is a $\calVS$-homomorphism: 
  For $i = 1 \ldots n$, $\ell \in \calLi$ and $t \in S$, we have both 
\begin{displaymath}
  \fsfn{\varepsilon_R}{\calRi}^{\calLi} ( \, \thetai(s) \,)(\myell)(\Rclass{t}) = \tssum_{t' \myin \Rclass{t}} \; \theta_\myi(s)(\myell)(t') 
  \quad \text{and} \quad 
  \theta^{\mkern3mu i}_R( \Rclass{s} )( \myell )( \Rclass{t} ) = \tssum_{t' \myin \Rclass{t}} \; \theta_\myi(s)(\myell)(t')
\end{displaymath}
  Thus, $\fsfn{\varepsilon_R}{\calRi}^\calLi \compose \theta_\myi = \theta^{\mkern3mu i}_R \compose \varepsilon_R$.
  Since $\calVS(\varepsilon_R) = \prod_{i=1}^n \: \fsfn{\varepsilon_R}{\calRi}^{\calLi}$ it follows that $\varepsilon_R$ is a $\calVS$-homo\-morphism.
  Therefore, by uniqueness of a final morphism, we have $\fmorphS{S}{\cdot} = \fmorphS{S/R}{\cdot} \compose \, \varepsilon_R$.
  In particular, $\fmorphS{S}{s_1} = \fmorphS{S}{s_2}$ since $\varepsilon_R(s_1) = \varepsilon_R(s_2)$. 
  Thus, $s_1 \approxS s_2$.
  
  For the reverse, i.e.\ $s_1 \approxS s_2 \IMPL s_1 \simS s_2$, assume $s_1 \approxS s_2$, i.e.\ $\fmorphS{S}{s_1} = \fmorphS{S}{s_2}$.
  Since the map $\fmorphS{S}{\cdot} : \twotuple{S}{\mytheta} \to \twotuple{\OmegaS}{\omega_{\calS}}$ is a $\calVS$-homomorphism, the relation~$R_{\calS}$ with $R_{\calS}(s',s'') \IFF \fmorphS{S}{s'} = \fmorphS{S}{s''}$ is an $\calS$-bisimulation:
  Suppose $R_{\calS}(s',s'')$, i.e.\ $s' \approxS s''$, for some $s',s'' \in \calS$.
  Assume $\theta_\OmegaS = \theta^1_\OmegaS \times \cdots \times \theta^{\mkern2mu n}_\OmegaS$.
  Pick $1 \leqslant i \leqslant n$, $\ell \in \calLi$, $t \in S$. 
  Put $\fmorphS{S}{t} = w \in \OmegaS$.
  Let $\Sclass{t}$ denote the equivalence class of~$t$ in~$R_{\calS}$.
\begin{displaymath}
\def\arraystretch{1.2}
  \begin{array}{rcll}
  \multicolumn{4}{l}{\tssum_{t' \in \Sclass{t}} \; \theta_\myi(s')(\myell)(t')}
  \\ & = &
  \tssum_{t' \in ( \, \fmorphS{S}{\cdot} \, )^{\mkern0mu -1} (w)} \; \theta_\myi (s')(\myell)(t')
  & \text{(by definition of~$R_{\calS}$ and~$w$)}
  \\ & = &
  \theta_{\OmegaS}^\myi ( \, \fmorphS{S}{s'} \mkern4mu \, )( \myell )( w )
  & \text{($\, \fmorphS{S}{\cdot}$ is a $\calVS$-homomorphism)}
  \\ & = &
  \theta_{\OmegaS}^\myi ( \, \fmorphS{S}{s''} \, )( \myell )( w )
  & \text{($s' \approxS s''$ by assumption)}
  \\ & = &
  \tssum_{t' \in ( \, \fmorphS{S}{\cdot} \, )^{\mkern0mu -1} (w)} \; \theta_\myi (s'')(\myell)(t')
  & \text{($\, \fmorphS{S}{\cdot}$ is a $\calVS$-homomorphism)}
  \\ & = &
  \tssum_{t' \in \Sclass{t}} \; \theta_\myi(s'')(\myell)(t') 
  & \text{(by definition of~$R_{\calS}$ and~$w$)}
  \end{array}
\def\arraystretch{1.0}
\end{displaymath}%
  Thus, if $R_{\calS}(s',s'')$ then $\tssum_{t' \in \Sclass{t}} \; \theta_\myi(s')(\myell)(t') = \tssum_{t' \in \Sclass{t}} \; \theta_\myi(s'')(\myell)(t')$ for all $i = 1 \ldots n$, $t \in S$, $\ell \in \calLi$ and $R_{\calS}$ is an $\calS$-bisimulation.
  Since $\fmorphS{S}{s_1} = \fmorphS{S}{s_2}$, it follows that $R_{\calS}(s_1,s_2)$.
  Thus $R_{\calS}$ is an $\calS$-bisimulation relating $s_1$ and~$s_2$.
  Conclusion, it holds that $s_1 \simS s_2$.
\end{proof}


\section{$\FuTS$ Semantics of $\PEPA$}
\label{sec-pepa}

  Next we will consider a significant fragment of the process algebra $\PEPA$~\cite{Hil96:phd}, including the parallel operator implementing the scheme of so-called minimal apparent rates, and provide a $\FuTS$ semantics for it.
  We will show that $\PEPA$'s notion of equivalence~$\eqpepa \mkern1mu$, called strong equivalence in~\cite{Hil96:phd}, fits with the bisimilarity~$\simpepa$ as arising from the $\FuTS$ semantics.

\blankline

\begin{definition}
  The set $\prcPEPA$ of $\PEPA$ processes is given by the BNF
\begin{math}
  P \bnfeq \nil \mid \alambda.P \mid P + P \mid P \prlA P \mid X
\end{math}
  where $a$~ranges over the set of actions~$\calA$, $\lambda$ over~$\poreals$, $A$~over the set of finite subsets of~$\calA$, and $X$~over the set of constants~$\calX$.
  \endef
\end{definition}

\blankline

\noindent
  $\PEPA$, like many other stochastic process algebras (e.g.\ \cite{HHM98:cnis,BG98:tcs}), couples actions and rates. 
  The prefix $\alambda$ of the process $\alambda.P$ expresses that the duration of the execution of the action~$a \in \calA$ is sampled from an exponential distribution of rate~$\lambda$.
  The parallel composition $P \prlA Q$ of a process~$P$ and a process~$Q$ for a set of actions~$A \subseteq \calA$ allows for the independent, asynchronous execution of actions of $P$ and~$Q$ not occurring in the subset~$A$, on the one hand, and requires the simultaneous, synchronized execution of $P$ and~$Q$ for the actions occurring in~$A$, on the other hand.
  The $\FuTS$-semantics of the fragment of $\PEPA$ that we consider here, is given by the SOS of Figure~\ref{fig-pepa-rules}, on which we comment below.  
  
  Characteristic for the $\PEPA$ language is the choice to model parallel composition, or cooperation in the terminology of $\PEPA$,  scaled by the minimum of the so-called apparent rates.
  By doing so, $\PEPA$'s strong equivalence becomes a congruence~\cite{Hil96:phd}.
  Intuitively, the apparent rate $r_a(P)$ of an action~$a$ for a process~$P$ is the sum of the rates of all possible $a$-executions for~$P$.
  When considering the CSP-style parallel composition $P \prlA Q$, with cooperation set~$A$, an action~$a$ occurring in~$A$ has to be performed by both $P$ and~$Q$.
  The rate of such an execution is governed by the slowest, on average, of the two processes in this respect.\footnote{One cannot take the slowest process per sample, because such an operation cannot be expressed as an exponential distribution in general.}
  Thus $r_a( \mkern1mu P \prlA Q \mkern1mu )$ for $a \in A$ is the minimum $\min \SET{ \, r_a(P), \, r_a(Q) \, }$.
  Now, if $P$~schedules an execution of~$a$ with rate~$r_1$ and $Q$~schedules a transition of~$a$ with rate~$r_2$, in the minimal apparent rate scheme the combined execution yields the action~$a$ with rate $r_1 \cdot r_2 \cdot \arf ( P , Q )$.
  Here, the `syntactic' scaling factor $\arf(P,Q)$, the apparent rate factor, is defined by
\begin{displaymath}
\def\arraystretch{1.2}
  \arf(P,Q) =
\begin{array}{c}
  \min \SET{ \, r_a(P), \, r_a(Q) \, }
  \\ \hline 
  r_a(P) \cdot r_a(Q)
\end{array}
\def\arraystretch{1.0}
\end{displaymath}
assuming $r_a(P), r_a(Q) > 0$, otherwise $\arf(P,Q) = 0$.
Thus, for $P \prlA Q$ the minimum $\min \SET{ \, r_a(P), \, r_a(Q) \, }$ of the apparent rates is adjusted by the relative probabilities $r_1/r_a(P)$ and~$r_2/r_a(Q)$ for executing~$a$ by~$P$ and~$Q$, respectively.
  See~\cite[Definition~3.3.1]{Hil96:phd} (or the appendix) for an explicit definition of the apparent rate~$r_a$ of a $\PEPA$-process.

\begin{figure}
\begin{displaymath}
\scalebox{0.90}{$
\begin{array}{c}
\sosrn{NIL}{}
\sosrule
  {\phantom{\mtrans{\alpha}_p}}
  {\nil \, \mtrans{\Edelaya}_p \, \zerof_{\nnreals}} 
\qquad
\sosrn{RAPF1}{}
\sosrule
  {\phantom{\mtrans{\alpha}_p}}
  {\alambda.P \, \mtrans{\Edelaya}_p \, [P \mapsto \lambda]}
\qquad
\sosrn{RAPF2}{}
\sosrule
  {b \neq a}
  {\alambda.P \, \mtrans{\Edelayb}_p \, \zerof_{\nnreals}} 
\bigskip \\
\sosrn{CHO}{}
\sosrule
  {P \, \mtrans{\Edelaya}_p \, \amset{P} \quad 
   Q \, \mtrans{\Edelaya}_p \, \amset{Q}}
  {P \cho Q \  \mtrans{\Edelaya}_p \  \amset{P} \cho \amset{Q}}
\qquad
\sosrn{CNS}{}
\sosrule
  {P \, \mtrans{\Edelaya}_p \, \amset{P} \quad X \dfas P}
  {X \, \mtrans{\Edelaya}_p \, \amset{P}}
  \bigskip \\
\sosrn{PAR1}{}
\sosrule
  {P \, \mtrans{\Edelaya}_p \, \amset{P} \quad
   Q \, \mtrans{\Edelaya}_p \, \amset{Q} \quad 
   a \, \notin \, A}
  {P \prlA Q \  \mtrans{\Edelaya}_p \ 
   ( \, \amset{P} \prlA \chut_Q \, ) 
   \, + \,
   ( \, \chut_P \prlA \amset{Q} \, )}
\qquad
\sosrn{PAR2}{}
\sosrule
  {P \, \mtrans{\Edelaya}_p \, \amset{P} \quad 
   Q \, \mtrans{\Edelaya}_p \, \amset{Q} \quad 
   a \, \in \, A}
  {P \prlA Q \  \mtrans{\Edelaya}_p \ 
   \arf( \mkern1mu {\amsetP}, {\amsetQ} \mkern1mu )
   \, \cdot \, ( \, \amset{P} \prlA \amset{Q} \, )}
\end{array}
$} 
\end{displaymath}
\halflineup
\halflineup
\caption{$\FuTS$ semantics for $\PEPA$.}
\label{fig-pepa-rules}
\end{figure}

\blankline

\noindent
  The $\FuTS$ we consider for the semantics of $\PEPA$ in Figure~\ref{fig-pepa-rules} involves a set of labels~$\Delta$ defined by $\Delta = \ZSET{\Edelaya}{a \in \calA }$.
  The symbol~$\Edelaya$ denotes the execution of the action~$a$, with a duration that is still to be established.
  The underlying semiring for the simple~$\FuTS$ for~$\PEPA$ is the semiring~$\nnreals$ of non-negative reals.
    
\begin{definition}
\label{df-ltfs-pepa}
  The $\FuTS$ $\calSpepa = \twotuple{\prcPEPA}{\mtrans{}_p}$ over $\Delta$ and $\nnreals$ has its transition relation given by the rules of Figure~\ref{fig-pepa-rules}.
  \endef
\end{definition}

\noindent
  We discuss the rules of Figure~\ref{fig-pepa-rules}.
  The $\FuTS$ semantics provides $\nil \mtrans{\Edelaya}_p \zerof_{\nnreals}$, for every action~$a$, with $\zerof_{\nnreals}$ the 0-function $\lambda P . 0$ of~$\nnreals$.
  However, the latter expresses $\thetapepa(\nil)(\Edelaya)(P') = 0$ for every $a \in \calA$ and $P' \in \prcPEPA$, or, in standard terminology, $\nil$ has no transition.
  For the rated action prefix $\alambda$ we distinguish two cases: (i)~execution of the prefix in rule~(RAPF1); (ii)~no~execution of the prefix in rule~(RAPF2).
  In the case of rule~(RAPF1) the label~$\Edelaya$ signifies that the transition involves the execution of the action~$a$.
  The continuation $[ \, P \mapsto \lambda \, ]$ is the function that assigns the rate~$\lambda$ to the process~$P$. 
  All other processes are assigned~$0$, i.e.\ the zero-element of the semiring~$\nnreals$. 
  In the second case, rule~(RAPF2), for labels~$\Edelayb$ with $b \neq a$, we do have a transition, but it is a degenerate one.
  The two rules for the prefix, in particular having the `null-continuation' rule (RAPF2), support the unified treatment of the choice operator in rule (CHO) and the parallel operator in rules (PAR1) and~(PAR2).
  
  Note the semantic sum of functions $\amset{P} \cho \amset{Q}$ replacing the syntactic sum in $P \cho Q$.
  The treatment of constants is as usual.
  Regarding the parallel operator~$\prlA$, with respect to some subset of actions $A \subseteq \calA$, the so-called cooperation set, there are again two rules.
  Now the distinction is between interleaving and synchronization.
  In the case of a label~$\Edelaya$ involving an action~$a$ not in the subset~$A$, either the $P$-operand or the $Q$-operand of $P \prlA Q$ makes progress.
  For example, the effect of the pattern $\amset{P} \prlA \chut_Q$ is that the value $\amset{P}(P') \cdot 1$ is assigned to a process~$P' \prlA Q$, the value $\amset{P}(P') \cdot 0 = 0$ to a process $P' \prlA Q'$ for some $Q' \neq Q$, and the value~$0$ for a process not of the form $P' \prlA Q'$.
  Here, as in all other rules, the right-hand sides of the transitions only involve functions in~$\fsfn{\prcPEPA}{\nnreals}$ and operators on them.
  
  For the synchronization case of the parallel construct, assuming $P \mtrans{\Edelaya} \amsetP$ and $Q \mtrans{\Edelaya} \amsetQ$, the `semantic' scaling factor $\arf(\amsetP,\amsetQ)$ is applied to~$\amsetP \prlA \amsetQ \,$ (with $\prlA$ on $\fsfn{\prcPEPA}{\nnreals}$ induced by~$\prlA$ on~$\prcPEPA$).
  This scaling factor, defined for functions in~$\fsfn{\prcPEPA}{\nnreals}$, is given by
\begin{displaymath}
  \arf( \mkern1mu \amsetP, \, \amsetQ \mkern1mu ) = 
  \begin{array}{c}
    \min \mkern2mu \SET{ \, \fsum \amsetP ,\, \fsum \amsetQ \, }
    \\ \hline
	\fsum \amsetP \cdot \fsum \amsetQ
  \end{array}
\end{displaymath}
provided $\fsum \amsetP, \fsum \amsetQ > 0$, and $\arf( \mkern1mu \amsetP, \, \amsetQ \mkern1mu ) = 0$ otherwise.  
  This results for $\arf( \mkern1mu {\amsetP}, {\amsetQ} \mkern1mu ) \, \cdot \, ( \, \amset{P} \prlA \amset{Q} \, )$, for a process~$R = R_1 \prlA R_2$, in the value $\arf ( \mkern1mu \amsetP, \, \amsetQ \mkern1mu ) \cdot  ( \, \amsetP \prlA \amsetQ  \, )( R_1 \prlA R_2) = \arf ( \mkern1mu \amsetP, \, \amsetQ \mkern1mu ) \cdot \amsetP(R_1) \cdot \amsetQ(R_2)$.
  

The following lemma establishes the relationship between the `syntactic' and `semantic' apparent rate factors defined on processes and on continuation functions, respectively.

\begin{lemma}
\label{lm-sem-syn-arf}
  Let $P \in \prcPEPA$ and $a \in \calA$. 
  Suppose $P \mtrans{\Edelaya}_p \amsetP$.
  Then $ r_a(P) = \fsum \amsetP$.
  \qed
\end{lemma}

\blankline

\noindent
  The proof of the lemma is straightforward.
  It is also easy to prove, by guarded induction, that the $\FuTS$ $\calSpepa$ given by Definition~\ref{df-ltfs-pepa} is total and deterministic.
  So, it is justified to write $\calSpepa = \twotuple{\prcPEPA}{\thetapepa}$. We use~$\simpepa$ to denote the bisimilarity induced by~$\calSpepa$.
  
\blankline

\begin{lemma}
\label{lm-pepa-total-det}
  The $\FuTS$ $\calSpepa$ is total and deterministic.
  \qed
\end{lemma}

\blankline

\begin{example}
  To illustrate the ease to deal with multiplicities in the $\FuTS$ semantics, consider the $\PEPA$ processes $P_1 = \alambda.P$ and $P_2 = \alambda.P + \alambda.P$ for some~$P \in \prcPEPA$.
  We have $P_1 \mtrans{\Edelaya}_p [ \, P \mapsto \lambda \, ]$ by rule (RAPF1), but $P_2 \mtrans{\Edelaya}_p [ \, P \mapsto 2 \lambda \, ]$ by rule (RAPF1) and rule~(CHO).
  The latter makes us to compute $[ \, P \mapsto \lambda \, ] + [ \, P \mapsto \lambda \, ]$, which equals $[ \, P \mapsto 2 \lambda \, ]$.
  Thus, in particular we have $P_1 \mathrel{\not{\sim}_{\pepa}} P_2$.  
  Intuitively it is clear that, in general we cannot have $P + P \sim P$ for any reasonable quantitative process equivalence~$\sim$ in the Markovian setting.
  Having twice as many $a$-labelled transitions, the average number for $\alambda.P + \alambda.P$ of executing the action~$a$ per time unit is double the average of executing~$a$ for~$\alambda.P$.
\end{example}

\blankline

\noindent
  The standard operational semantics of~$\PEPA$~\cite{Hil96:phd,Hil05:lics} is given in Figure~\ref{fig-standard-pepa-rules}.
  The transition relation ${\trans{}} \subseteq \prcPEPA \times ( \, \calA \times \poreals \, ) \times \prcPEPA$ is the least relation satisfying the rules.
  For a proper treatment of the rates, the transition relation is considered as a multi-transition system, where also the number of possible derivations of a transition $P \trans{a,\lambda} P'$ matters.
  We stress that such bookkeeping is not needed in the $\FuTS$-approach at all.
  In rule~(PAR2) we use the `syntactic' apparent rate factor for $\PEPA$ processes.
\begin{figure}
\begin{displaymath}
\scalebox{0.90}{$
\begin{array}{c}
\sosrn{RAPF}{}
\sosrule
  {\phantom{\mtrans{\alpha}_p}}
  {\alambda.P \, \trans{a,\lambda} P}
\qquad
\sosrn{CHO1}{}
\sosrule
  {P \, \trans{a,\lambda} \, P'} 
  {P \cho Q \  \trans{a,\lambda} P'} 
\qquad
\sosrn{CHO2}{}
\sosrule
  {Q \, \trans{a,\lambda} \, Q'} 
  {P \cho Q \  \trans{a,\lambda} P'} 
\bigskip \\
\sosrn{PAR1a}{}
\sosrule
  {P \, \trans{a,\lambda} \, P' \quad 
   a \, \notin \, A}
  {P \prlA Q \  \trans{a,\lambda} P' \prlA Q}
\qquad
\sosrn{PAR1b}{}
\sosrule
  {Q \, \trans{a,\lambda} Q' \quad 
   a \, \notin \, A}
  {P \prlA Q \  \trans{a,\lambda} \  P \prlA Q'}
\qquad
\sosrn{CNS}{}
\sosrule
  {P \, \trans{a,\lambda} \, P' \quad X \dfas P}
  {X \, \trans{a,\lambda} \, P'}
\bigskip \\
\sosrn{PAR2}{}
\sosrule
  {P \, \trans{a,\lambda_1} \, P' \quad 
   Q \, \trans{a,\lambda_2} \, Q' \quad 
   a \, \in \, A}
  {P \prlA Q \  \trans{a,\lambda} \ P' \prlA Q'}
   \quad \text{$\lambda = \arf( \mkern1mu {P}, {Q} \mkern1mu ) {\cdot} \lambda_1 {\cdot} \lambda_2$}
\end{array}
$} 
\end{displaymath}
\halflineup
\halflineup
\caption{Standard semantics for $\PEPA$.}
\label{fig-standard-pepa-rules}
\end{figure}

\blankline

\noindent
  The so-called total conditional transition rate $q[P,C,a]$ of a $\PEPA$-process~\cite{Hil96:phd,Hil05:lics} for a subset of processes $C \subseteq \prcPEPA$ and~$a \in \calA$ is given by 
\begin{math}
  q[P,C,a] 
  = 
  \tssum_{Q \in C} \; \sum \MSET{ \lambda }{P \trans{a,\lambda} Q}
\end{math}.
  Here, $\mlbrace \, P \trans{a,\lambda} Q \, \mrbrace$ is the multiset of  transitions $P \trans{a,\lambda} Q$ and $\MSET{ \lambda }{P \trans{a,\lambda} Q}$ is the multiset of all~$\lambda$'s involved.
  The multiplicity of $P \trans{a,\lambda} Q$ is the number of different ways the transition can be derived using the rules of Figure~\ref{fig-standard-pepa-rules}. 
  We are now ready to define $\PEPA$'s notion of strong equivalence~\cite{Hil96:phd,Hil05:lics}.
  
\blankline

\begin{definition}
\label{df-strong-equiv}
  An equivalence relation $R \subseteq \prcPEPA \times \prcPEPA$ is called a strong equivalence if $q[ P_1 , \Rclass{Q} , a ] = q[ P_2 , \Rclass{Q} , a ]$ for all $P_1, P_2 \in \prcPEPA$ such that $R(P_1,P_2)$, all~$Q \in \prcPEPA$ and all~$a \in \calA$.
  Two processes $P_1, P_2 \in \prcPEPA$ are strongly equivalent if $R(P_1,P_2)$ for a strong equivalence~$R$, notation $P_1 \eqpepa P_2$.
  \endef
\end{definition}

\blankline

\noindent
  The next lemma couples, for a $\PEPA$-process~$P$, an action~$a$ and a function $\amsetP \in \fsfn{\prcPEPA}{\nnreals}$, the evaluation~$\amsetP(P')$ with respect to the $\FuTS$-semantics to the cumulative rate for~$P$ of reaching~$P'$ by a transition involving the label~$a$ in the standard operational semantics.

\blankline

\begin{lemma}
\label{lm-cnt-ltfs-match}
  Let $P \in \prcPEPA$ and $a \in \calA$. 
  Suppose $P \mtrans{\Edelaya} \amsetP$.
  Then it holds that $\amsetP (P') = \tssum \; \MSET{ \lambda }{ P \trans{a,\lambda} P' }$ for all~$P' \in \prcPEPA$.
  \qed
\end{lemma}

\blankline

\noindent
With the lemma in place we can prove the following correspondence result for $\calSpepa$-bisimilarity with respect to the $\FuTS$ for~$\PEPA$ of Definition~\ref{df-ltfs-pepa} and strong equivalence as given by Definition~\ref{df-strong-equiv}.

\blankline

\begin{theorem}
\label{th-pepa-strong-equiv-lfts-bisim}
  For any two $\PEPA$-processes $P_1, P_2 \in \prcPEPA$ it holds that $P_1 \simpepa P_2$ iff $P_1 \eqpepa P_2$.
\end{theorem}
\begin{proof}
  Let~$R$ be an equivalence relation on~$\prcPEPA$.
  Choose $P,Q \in \prcPEPA$ and~$a \in \calA$.
  Suppose $P \mtrans{\Edelaya}_p \amsetP$.
  Thus $\theta_\pepa (P)(\Edelaya) = \amsetP$.
  We have 
\begin{displaymath}
\def\arraystretch{1.2}
\begin{array}[t]{rcll}
  q[ P, \Rclass{Q}, a ]
  & = &
  \tssum_{Q' \in \Rclass{Q}} \; \tssum\MSET{ \lambda }{P \trans{a,\lambda} Q'}
  & \text{(by definition $q[ P, \Rclass{Q}, a ]$}
  \\ & = &
  \tssum_{Q' \in \Rclass{Q}} \; \amsetP( Q' ) 
  & \text{(by Lemma~\ref{lm-cnt-ltfs-match})}
  \\ & = &
  \tssum_{Q' \in \Rclass{Q}} \; \thetapepa(P)(a)(Q')
  & \text{(by definition $\thetapepa$)}
\end{array}
\def\arraystretch{1.0}
\end{displaymath} 
  Therefore, for $\PEPA$-processes $P_1$ and~$P_2$ it holds that 
$q[ P_1, \Rclass{Q}, a ] = q[ P_2, \Rclass{Q}, a ]$ for all $Q \in \prcPEPA$, $a \in \calA$ iff $\tssum_{Q' \in \Rclass{Q}} \; \thetapepa(P_1)(a)(Q') = \tssum_{Q' \in \Rclass{Q}} \; \thetapepa(P_2)(a)(Q')$ for all $Q \in \prcPEPA$, $a \in \calA$.
  Thus, the equivalence relation~$R$ is a strong equivalence iff $R$~is an $\calSpepa$-bisimulation, from which the theorem follows.
\end{proof}

\noindent
In view of our general correspondence result Theorem~\ref{th-correspondence}, the above theorem shows that $\PEPA$'s strong equivalence~$\eqpepa$ is a behavioral equivalence, viz.\ behavioral equivalence~$\approx_\pepa$ with respect to the functor of $\calSpepa$, and that its standard, $\FuTS$ and coalgebraic semantics coincide.


\section{$\FuTS$ Semantics of $\IML$}
\label{sec-iml}

In this section we provide a $\FuTS$ semantics for a relevant part of the language of $\IMC$~\cite{Her02:springer}.
  $\IMC$,~Interactive Markov Chains, are automata that combine two types of transitions: interactive transitions that involve the execution of actions and Markovian transitions that represent the progress of time governed by exponential distribution.
  As a consequence, $\IMC$ embody both non-deterministic and stochastic behaviour.
  System analysis using $\IMC$ proves to be a powerful approach because of the orthogonality of qualitative and quantitative dynamics, their logical underpinning and tool support.
  A number of equivalences, both strong and weak, are available for $\IMC$~\cite{EHZ10:concur}.
  In our treatment here, dealing with a fragment we call~$\IML$, we do not deal with internal $\tau$-steps and focus on strong bisimulation.

\blankline

\begin{definition}
  The set $\prcIML$ of $\IML$ processes is given by the BNF
\begin{math}
  P \bnfeq \nil \mid a.P \mid \lambda.P \mid P + P \mid P \prlA P \mid X
\end{math}
  where $a$~ranges over the set of actions~$\calA$, $\lambda$~over~$\poreals$, $A$~over the set of finite subsets of~$\calA$ and $X$~over the set of constants~$\calX$.
  \endef
\end{definition}

\blankline

\noindent
  In~$\IML$ there are separate prefix constructions for actions~$a.P$ and for time-delays~$\lambda.P$.
  No restriction is imposed on the alternative and parallel composition of processes.
  For example, we have the process $P = a.\lambda.\nil \cho \mu.b.\nil$ in~$\IML$.
  It should be noted that for~$\IMC$ actions are considered to take no time. 

\blankline

\begin{definition}
\label{df-iml-sem}
  The formal semantics of $\prcIML$ is given by the $\FuTS$
  $\calSiml = \threetuple{\, \prcIML}{\mtrans{}_1}{\mtrans{}_2}$ over the label sets $\calA$ and~$\Delta = \SET{ \mkern2mu \delta \mkern2mu }$ and the semirings $\bools$ and~$\nnreals$ with transition relations ${\mtrans{}_1} \subseteq \prcIML \times \calA \times \fsfn{\prcIML}{\bools}$ and ${\mtrans{}_2}\subseteq \prcIML \times \Delta \times \fsfn{\prcIML}{\nnreals}$ defined as the least relations satisfying the rules of Figure~\ref{fig-iml-rules}. 
\endef
\end{definition}

\begin{figure}
\begin{displaymath}
\scalebox{0.85}{$
\begin{array}{c}

\sosrn{NIL1}{}
\sosrule
  {a \in \calA}
  {\nil \, \mtrans{a}_1 \, \zerof_\bools} 
  
\qquad

\sosrn{NIL2}{}
\sosrule
  {}
  {\nil \, \mtrans{\Edelay}_2 \, \zerof_{\nnreals}}
  
\qquad

\sosrn{APF3}{}
\sosrule
  {}
  {a.P \, \mtrans{\Edelay}_2 \, \zerof_{\nnreals}} 

\bigskip \\

\sosrn{APF1}{}
\sosrule
  {}
  {a.P \, \mtrans{a}_1 \, [P \mapsto \TRUE]}
  
\quad

\sosrn{APF2}{}
\sosrule
  {b \neq a}
  {a.P \, \mtrans{b}_1 \, \zerof_\bools}
  
\quad

\sosrn{RPF1}{}
\sosrule
  {a \in \calA}
  {\lambda.P \, \mtrans{a}_1 \, \zerof_\bools}
  
\quad

\sosrn{RPF2}{}
\sosrule
  {}
  {\lambda.P \, \mtrans{\Edelay}_2 \, [P \mapsto \lambda]}

\bigskip \\

\sosrn{PAR1}{}
\sosrule
  {P \, \mtrans{\alpha}_i \, \amset{P} \quad
   Q \, \mtrans{\alpha}_i \, \amset{Q} \quad 
   \alpha \notin A}
  {P \prlA Q \  \mtrans{\alpha}_i \ 
   ( \, \amset{P} \prlA \chut^{\mkern2mu i}_{\mkern-2mu Q} \, ) 
   \, + \,
   ( \, \chut^{\mkern2mu i}_{\mkern-2mu P} \prlA \amset{Q} \, )}
  (i=1,2)
  
\qquad

\sosrn{PAR2}{}
\sosrule
  {P \, \mtrans{a}_1 \, \amset{P} \quad 
   Q \, \mtrans{a}_1 \, \amset{Q} \quad 
   a \in A}
  {P \prlA Q \  \mtrans{a}_1 \ 
   \amset{P} \prlA \amset{Q}}
   
\bigskip \\

\sosrn{CHO}{}
\sosrule
  {P \, \mtrans{\alpha}_i \, \amset{P} \quad 
   Q \, \mtrans{\alpha}_i \, \amset{Q}}
  {P \cho Q \  \mtrans{\alpha}_i \  \amset{P} \cho \amset{Q}} 
  (i=1,2)

\qquad

\sosrn{CON}{}
\sosrule
  {P \, \mtrans{\alpha}_i \, \amset{P} \quad 
   X \dfas P}
  {X \, \mtrans{\alpha}_i \, \amset{P}}
  (i=1,2)

\end{array}
$} 
\end{displaymath}
\halflineup
\halflineup
\caption{$\FuTS$ semantics for $\IML$.}
\label{fig-iml-rules}
\end{figure}

\blankline

\noindent
  To accommodate for action-based and delay-related transitions, the $\FuTS$~$\calSiml$ is non-simple, having the two transition-to-function relations $\mtrans{}_1$ and~$\mtrans{}_2$.
  Actions~$a \in \calA$ decorate~$\mtrans{}_1$, the special symbol~$\delta$ decorates~$\mtrans{}_2$.
  Note rule~(APF3) and rule~(RPF1) that involve the null-functions of~$\nnreals$ and of~$\bools$, respectively, to express that a process~$a.P$ does not trigger a delay and a process~$\lambda.P$ does not execute an action.
  For the parallel construct~$\prlA$, interleaving applies both for non-synchronized actions~$a \notin A$ as well as for delays (but not mixed). 
  Therefore, rule~(PAR1) pertains to both $\mtrans{}_1$ and~$\mtrans{}_2$, with~$\alpha$ ranging over~$\calA \cup \Delta$.
  The same holds for non-deterministic choice, rule~(CHO), and constants, rule~(CON).
  Finally, $\IML$ does not provide synchronization of delays in the parallel construct.
  Rule~(PAR2) only concerns the transition relation~$\mtrans{}_2$.
  In rule~(PAR1), for clarity, we decorated the characteristic functions, writing $\chut^\myi_{\mkern-1mu P} \mkern1mu$, for $i=1,2$, for $\chut_{\mkern-1mu P} = [ \, P \mapsto \TRUE \, ]$ in~$\fsfn{ \prcIML }{ \bools }$ and $\chut_{\mkern-1mu P} = [ \, P \mapsto 1 \, ]$ in~$\fsfn{ \prcIML }{ \nnreals }$. 
  
\blankline

\begin{example}
  Assume $X \dfas a.\lambda.b.X$ and $Y \dfas a.\mu.b.Y$.
  Put $A = \SET{a,b}$.
  Then we have
\begin{displaymath}
\scalebox{.99}{$
\begin{array}{r@{\, \prlA \,}l@{\;}c@{\;}l@{\,}r@{\, \prlA \,}l@{\,}lcr@{\, \prlA \,}l@{\;}c@{\;}l@{\,}r@{\, \prlA \,}l@{\;}l}
  X & Y & \mtrans{a}_1 & 
  [ & \lambda.b.X & \mu.b.Y & \mapsto \TRUE \, ] 
  & \  &
  \lambda.b.X & \mu.b.Y & \mtrans{\Edelay}_2 &
  [ &  b.X & \mu.b.Y & \mapsto \lambda ,\ 
  \lambda.b.X \prlA b.Y \mapsto \mu \, ]
  \\
  b.X & b.Y & \mtrans{b}_1 &
  [ & X & Y & \mapsto \TRUE \, ]
  && b.X & \mu.b.Y & \mtrans{\Edelay}_2 &
  [ & b.X & b.Y & \mapsto \mu \, ]
  \\ \multicolumn{5}{c}{} &&&&
  \lambda.b.X & b.Y & \mtrans{\Edelay}_2 &
  [ & b.X & b.Y & \mapsto \lambda \, ]
\end{array}$}
\end{displaymath}
\end{example}

\blankline

\noindent
  It is not difficult to verify that $\calSiml$ is a total and deterministic $\FuTS$.
  Below we use $\calSiml = \threetuple {\prcIML}{\theta_1}{\theta_2}$ and write $\simiml$ for the associated bisimilarity.
  
\blankline

\begin{lemma}
\label{lm-iml-total-det}
  The $\FuTS$ $\calSiml$ is total and deterministic.
  \qed
\end{lemma}

\blankline

\noindent
  The standard SOS semantics of~$\IML$~\cite{Her02:springer} is given in Figure~\ref{fig-standard-iml} involving the transition relations 
\begin{displaymath}
  {\trans{}} \subseteq \prcIML \times \calA \times \prcIML 
  \qquad \text{and} \qquad 
  {\Mtrans{}} \subseteq \prcIML \times \poreals \times \prcIML
\end{displaymath}
  Below we will use the functions $\bfT$ and~$\bfR$ based on $\trans{}$ and~$\Mtrans{}$, cf.~\cite{HK10:fmco}.
  We have $\bfT \colon \prcIML \times \calA \times {\textbf{2}}^{\prcIML} \to \bools$ given by $\bfT( P, a, C ) = \TRUE$ if the set $\ZSET{ P' \in C }{ P \trans{a} P' }$ is non-empty, for all $P \in \prcIML$, $a \in \calA$ and any subset~$C \subseteq \prcIML$.
  For $\bfR \colon \prcIML \times \prcIML \to \nnreals$ we put $\bfR(P,P') = \tssum \MSET{ \lambda }{ P \Mtrans{\lambda} P' }$.
  Here, as common for probabilistic and stochastic process algebras, the comprehension is over the multiset of transitions leading from~$P$ to~$P'$ with label~$\lambda$.
  We extend $\bfR$ to $\prcIML \times {\textbf{2}}^{\prcIML}$ by $\bfR(P,C) = \tssum_{P' \myin C} \; \tssum \MSET{ \lambda }{ P \Mtrans{\lambda} P' }$.
  
  For $\IML$ we have the following notion of strong bisimulation~\cite{Her02:springer,HK10:fmco} that we will compare with the notion of bisimulation associated with the $\FuTS$~$\calSiml$.
 
\begin{figure}
\begin{displaymath}
\scalebox{0.85}{$
\begin{array}{c}

\sosrn{APF}{}
\sosrule
  {}
  {\pfx{a}{P} \, \trans{a} \, P}

\qquad

\sosrn{CHO1}{}
\sosrule
  {P \, \trans{a} \, R}
  {P \cho Q  \, \trans{a} \, R}

\qquad

\sosrn{CHO2}{}
\sosrule
  {Q \, \trans{a} \, R}
  {P \cho Q  \, \trans{a} \, R}

\qquad
 
\sosrn{CON1}{}
\sosrule
  {P \, \trans{a} \, Q \quad X \dfas P}
  {X \, \trans{a} \, Q}
  
\bigskip \\

\sosrn{PAR1a}{}
\sosrule
  {P \, \trans{a} \, P' \quad a \notin A}
  {P \prlA Q \, \trans{a}\ , P' \prlA Q}

\qquad

\sosrn{PAR1b}{}
\sosrule
  {Q \, \trans{a} \, Q' \quad a \notin A}
  {P \prlA Q \, \trans{a} \, P \prlA Q'}

\qquad

\sosrn{PAR2}{}
\sosrule
  {P \, \trans{a} \, P' \quad 
   Q \, \trans{a} \, Q' \quad a \in A}
  {P \prlA Q \, \trans{a} \, P' \prlA Q'}

\displaybreak[3] 
\bigskip \\

\sosrn{RPF}{}
\sosrule{}
  {\pfx{\lambda}{P} \, \Mtrans{\lambda} \, P}

\qquad

\sosrn{CHO3}{}
\sosrule
  {P \, \Mtrans{\lambda} \, R}
  {P \cho Q  \, \Mtrans{\lambda} \, R}

\qquad

\sosrn{CHO4}{}
\sosrule
  {Q \, \Mtrans{\lambda} \, R}
  {P \cho Q \, \Mtrans{\lambda} \, R}

\qquad

\sosrn{CON2}{}
\sosrule
  {P \,  \Mtrans{\lambda} \, Q \quad X \dfas P}
  {X \, \Mtrans{\lambda} \, Q}
  
\bigskip \\

\sosrn{PAR1c}{}
\sosrule
  {P \, \Mtrans{\lambda} \, P'}
  {P \prlA Q \, \Mtrans{\lambda} \, P' \prlA Q}

\qquad

\sosrn{PAR1d}{}
\sosrule
  {Q \, \Mtrans{\lambda} \, Q'}
  {P \prlA Q \, \Mtrans{\lambda} \, P \prlA Q'}
\end{array}
$} 
\end{displaymath}
\halflineup
\halflineup
\caption{Standard SOS rules for $\IML$.}
\label{fig-standard-iml}
\end{figure}

\blankline

\begin{definition}
\label{df-iml-strong-bisimulation}
  An equivalence relation $R \subseteq \prcIML \times \prcIML$ is called a strong bisimulation for~$\IML$ if, for all $P_1, P_2 \in \prcIML$ it holds that
\begin{itemize}
  \item for all $a \in \calA$ and $Q \in \prcIML$\emph{:} $\bfT \threetuple {P_1} a {\Rclass{Q}} \iff \bfT \threetuple {P_2} a {\Rclass{Q}}$
  \smallskip
  \item for all $Q \in \prcIML$\emph{:} $\bfR \twotuple {P_1}{\Rclass{Q}} \; = \; \bfR \twotuple {P_2}{\Rclass{Q}}$.
\end{itemize}
  for all $P_1, P_2 \in \prcIML$ such that $R(P_1,P_2)$.
  Two processes $P_1, P_2 \in \prcIML$ are called strongly bisimilar if $R(P_1,P_2)$ for a strong bisimulation~$R$ for~$\IML$, notation $P_1 \eqiml P_2$.
  \endef
\end{definition}

\blankline

\noindent
  To establish the correspondence of $\FuTS$ bisimilarity~$\simiml$ for~$\calSiml$ of Definition~\ref{df-iml-sem} and strong bisimilarity~$\eqiml$ for~$\IML$, we need to connect the state-to-function relation~$\mtrans{}_1$ and the transition relation~$\trans{}$ as well as the state-to-function relation~$\mtrans{}_2$ and the transition relation~$\Mtrans{} \,$.

\blankline

\begin{lemma}
\label{lm-mtrans-trans}
\begin{itemize}
\item [] \mbox{}
\item [(a)]
  Let $P \in \prcIML$ and $a \in \calA$. 
  If $P \mtrans{a}_1 \amsetP$ then $P \trans{a} P' \iff \amsetP( P' ) = \TRUE$.
\item [(b)]
  Let $P \in \prcIML$ . 
  If $P \mtrans{\Edelay}_2 \amsetP$ then $\tssum \MSET{ \lambda }{ P \Mtrans{\lambda} P' } = \amsetP( P' )$.
  \qed
\end{itemize}
\end{lemma} 

\blankline

\noindent
We are now in a position to relate $\FuTS$ bisimulation and standard strong bisimulation for~$\IML$.

\blankline

\begin{theorem}
\label{th-correpondence-iml}
  For any two processes $P_1, P_2 \in \prcIML$ it holds that $P_1 \simiml P_2$ iff $P_1 \eqiml P_2$.
\end{theorem} 
\begin{proof}
  Let~$R$ be an equivalence relation on~$\prcIML$.
  Pick $P \in \prcIML$, $a \in \calA$ and choose any~$Q \in \prcIML$.
  Suppose $P \mtrans{a} \amsetP$. 
  Thus $\theta_1(P)(a) = \amsetP$.
  Then we have
\begin{displaymath} 
\def\arraystretch{1.2}
\begin{array}{rcll}
  \bfT \threetuple {P} a {\Rclass{Q}}  
  & \Leftrightarrow &
  \exists Q' \in \Rclass{Q} \colon P \trans{a} Q'
  & \text{(by definition of $\bfT$)}
  \\ & \Leftrightarrow &
  \exists Q' \in \Rclass{Q} \colon \amsetP(Q') = \TRUE
  & \text{(by Lemma~\ref{lm-mtrans-trans}a)}
  \\ & \Leftrightarrow &
  \tssum_{Q' \in \Rclass{Q}} \; \theta_1(P)(a)(Q) = \TRUE
  & \text{(by definition of $\theta_1$)}
\end{array}
\def\arraystretch{1.0}
\end{displaymath}
  Note, summation in~$\bools$ is disjunction.
  Likewise, on the quantitative side, we have
\begin{displaymath} 
\def\arraystretch{1.2}
\begin{array}{rcll}
  \bfR \twotuple {P}{\Rclass{Q}}  
  & = &
  \tssum_{Q' \in \Rclass{Q}} \; \tssum \MSET{ \lambda }{ P \Mtrans{\lambda} Q' }
  & \text{(by definition of $\bfR$)}
  \\ & = &
  \tssum_{Q' \in \Rclass{Q}} \; \amsetP(Q' \mkern1mu ) 
  & \text{(by Lemma~\ref{lm-mtrans-trans}b)}
  \\ & = &
  \tssum_{Q' \in \Rclass{Q}} \; \theta_2(P)(\Edelay)(Q)
  & \text{(by definition of $\theta_2$)}
\end{array}
\def\arraystretch{1.0}
\end{displaymath}
  Combining the equations, we conclude that a strong bisimulation for~$\IML$ is also a bisimulation for the $\FuTS$~$\calSiml$, and vice versa.
  From this the theorem follows.
\end{proof}

\noindent
Again, as a corollary of the theorem above, we have for~$\IML$ that its notion of strong bisimulation is coalgebraically underpinned, as it coincides, calling to Theorem~\ref{th-correspondence} once more, with behavioral equivalence of the functor~$\calViml$ induced by the $\FuTS$~$\calSiml$.
  As a consequence, the standard, $\FuTS$ and coalgebraic semantics for~$\IML$ are all equal.


\section{Concluding remarks}
\label{sec-conclusions}

Total and deterministic labeled state-to-function transition systems, $\FuTS$, are a convenient instrument to express the operational semantics of both qualitative and quantitative process languages.
  In this paper we have introduced the notion of bisimulation that arises from a $\FuTS$, possibly involving multiple transition relations.
  A correspondence result, Theorem~\ref{th-correspondence}, relates the bisimulation of a $\FuTS$~$\calS$ to behavioral equivalence of the functor $\calVS$ that arises from the $\FuTS$~$\calS$ too.
  For two prototypical stochastic process languages based on $\PEPA$ and on~$\IMC$ we have shown that the notion of stochastic bisimulation associated with these calculi, coincides with the notion of bisimulation of the corresponding~$\FuTS$.\
  Using these $\FuTS$ as a stepping stone, the correspondence result bridges between the concrete notion of bisimulation for $\PEPA$ and~$\IMC$, and the coalgebraic notion of behavioral equivalence.
  Hence, from this perspective, the concrete notions are seen as the natural strong equivalence to consider.

  It is shown in~\cite{BBBRS11}, in the context of weighted automata, that in general the type of functors~$\fsfn{{\cdot}}{\calR}$ may not preserve weak pullbacks and, therefore, the notions of coalgebraic bisimulation and of behavioral equivalence may not coincide.
  Essential for the construction in their setting is the fact that the sum of non-zero weights may add to weight~$0$.
  The same phenomenon prevents a general proof, along the lines of~\cite{DeR99}, for coalgebraic bisimulation and $\FuTS$ bisimulation to coincide.
  In the construction of a mediating morphism, going from $\FuTS$ bisimulation to coalgebraic bisimulation a denominator may be zero, hence a division undefined, in case the sum over an equivalence class cancels out.
  In the concrete case for~\cite{KS08:fossacs}, although no detailed proof is provided there, this will not happen with~$\nnreals$ as underlying semiring.
  We expect that for semirings enjoying the property that for a sum $x = x_1 + \cdots + x_n$ it holds that $x = 0$ iff $x_i = 0$ for all $i = 1 \ldots n$, we will be able to prove that pullbacks are weakly preserved, and hence that coalgebraic bisimulation and behavioral equivalence are the same.
  
  Obviously, Milner-type strong bisimulation~\cite{Mil80:lncs,Par81:gi} and bisimulation for $\FuTS$ over~$\bools$ coincide. 
  Also, strong bisimulation of~\cite{Hil96:phd} involving, apart from the usual transfer conditions, the comparison of state information, viz.\ the apparent rates, can be treated with~$\FuTS$. 
  Again the two notions of equivalence coincide.
  We expect to be able to deal with discrete time and so-called Markov automata as well.
  For dense time and general measures one may speculate that the use of functions of compact support with respect to a suitable topology may be fruitful.
  Future research needs to reveal under what algebraic conditions of the semirings, or similar structures, or the coalgebraic conditions on the format of the functors involved standard bisimulation, $\FuTS$-bisimulation, coalgebraic bisimulation and behavioral equivalence will amount to similar identifications.
  
\emph{Acknowledgments}
  The authors are grateful to Rocco De~Nicola, Michele Loreti and Jan Rutten for fruitful discussions useful suggestions.
  DL and MM acknowledge support by EU Project n.~257414 Autonomic Service-Components Ensembles (ASCENS) and by CNR/RSTL Project~XXL\@.
  This research has been conducted while EV was spending a sabbatical leave at the CNR/ISTI\@.
  EV gratefully acknowledges the hospitality and support during his stay in Pisa.

\halflineup

\bibliographystyle{eptcs}
\bibliography{accat}

\vfill{}
\pagebreak


\appendix

\section{Additional proofs}

\subsection*{Additional proof for Section~\ref{sec-preliminaries}}

\myblock{Lemma~\ref{lm-props-fsum}}
  Let~$X$ be a set, $\calR$ a semiring and $\myop$ an injective binary operation on~$X$.
  For $\varphi, \psi \in \fsfn{X}{\calR}$ it holds that $\fsum (\varphi \cho \psi ) = \fsum \varphi \cho \fsum \psi$ and $\fsum ( \, \varphi \myop \psi \, ) = ( \fsum \varphi ) \ast ( \fsum \psi )$.
\begin{proof}
  We verify $\fsum ( \, \varphi \myop \psi \, ) = ( \fsum \varphi ) \ast ( \fsum \psi )$ for arbitrary $\varphi, \psi \in \fsfn{X}{\calR}$.
\begin{displaymath}
\def\arraystretch{1.2}
\begin{array}{rcll}
  \multicolumn{3}{l}{%
  ( \fsum \varphi ) \ast ( \fsum \psi )
  } 
  \\ & = &
  \bigl( \, \tssum_{x_1 \myin X} \; \varphi(x_1) \, \bigr) 
  \ast
  \bigl( \, \tssum_{x_2 \myin X} \; \psi(x_2) \, \bigr) 
  & \text{(by definition $\fsum$)}
  \\ & = &
  \tssum_{x_1, x_2 \myin X} \; \varphi(x_1) \ast \psi(x_2)
  & \text{(by distributivity of $\ast$ over~$+$)}
  \\ & = &
  \tssum_{x \myin X ,\, x = x_1 \myop x_2} \; \varphi(x_1) \ast \psi(x_2)
  & \text{(by injectivity of $\myop \,$)}
  \\ & = &
  \tssum_{x \myin X ,\, x = x_1 \myop x_2} \; ( \varphi \myop \psi )(x)
  & \text{(by definition of $\varphi \myop \psi$)}
  \\ & = &
  \tssum_{x \myin X } \; ( \varphi \myop \psi )(x)
  & \text{(since $( \varphi \myop \psi )(x) = 0$ if $x \notin \myop^{-1}(X)$)}
\end{array}
\def\arraystretch{1.0}
\end{displaymath}
The fact $\fsum (\varphi \cho \psi ) = \fsum \varphi \cho \fsum \psi$ follows direct from the definitions and commutativity of~$+$.
\end{proof}

\subsection*{Additional proof for Section~\ref{sec-coalgebra}}

\myblock{Lemma~\ref{lm-v-is-bounded}}
  Let~$\calL$ be a set of labels and $\calR$~a semiring.
  Then functor~$\calVLR$ on~$\Set$ is bounded.

\begin{proof}
  Consider the elements $\nu \in \fsfn{\nats}{\calR}^\calL$ as parametrized `valuation' functions, and the elements $\sigma \in X^\nats$ as `selection' functions. 
  The functor $\fsfn{\nats}{\calR}^\calL  \times ( {\cdot} )^\nats : \Set \to \Set$ is the product functor of the lifting $\mathit{id}_{\fsfn{\nats}{\calR}^\calL}$ of the identity functor $\mathit{id}_{\fsfn{\nats}{\calR}}$ to~$\calL$ and of the functor $( {\cdot} )^\nats$.
  Define the mapping  $\eta : \mathit{id}_{\fsfn{\nats}{\calR}^\calL} \times ( {\cdot} )^\nats \to \calVLR$ by putting 
\begin{displaymath}
  \eta_X( \nu , \sigma )(\myell)(x) \, = \, 
  \tssum_{\, n \myin \sigmainv(x)} \; \nu \mkern1mu (\myell)(n)
\end{displaymath}
for $\nu \in \fsfn{\nats}{\calR}^\calL$, $\sigma \in X^\nats$ and $\ell \in \calL$. 
  For $\ell \in \calL$ and $x \in X$, the right-hand sum defining $\eta_X( \nu , \sigma )(\myell)(x)$ exists, since $\nu \mkern1mu (\myell) : \nats \to \calR$ is of finite support.
  Note that $\eta_X( \nu , \sigma )( \myell )$ is of finite support too:
  If $\eta_X(\nu,\sigma)(\myell)(x) \linebreak  \neq 0$, by definition $\sum \ZSET{ \nu \mkern1mu (\myell)(n) }{ n \in \sigmainv (x) } \neq 0$. 
  Then $\nu \mkern1mu (\myell)(n) \neq 0$ for some $n \in \sigmainv(x)$. 
  Thus, $\sigma(n) = x$ for some $n \in \spt(\nu \mkern1mu (\myell) )$.
  So, $\spt( \, \eta_X( \nu , \sigma ) \, ) \subseteq \ZSET{ \sigma(n) }{ n \in \spt( \nu  \mkern1mu (\myell) ) }$ and $\spt( \, \eta_X( \nu , \sigma ) \, )$ is finite.
  
  Next we verify that $\eta : \mathit{id}_{\fsfn{\nats}{\calR}^\calL} \times ( {\cdot} )^\nats \to \calVLR$ is a natural transformation, i.e.\ we check that for $f : X \to Y$ it holds that $\calVLR \compose \eta_X = \eta_Y \compose \la \, \mathit{id}_{\fsfn{\nats}{\calR}^\calL} ,\, f^\nats \ra$.
  
  \smallskip
  \centerline{\includegraphics[scale=0.60]{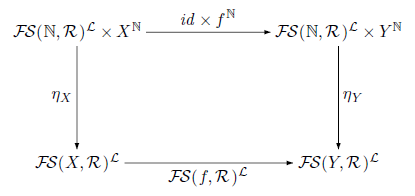}}
  \smallskip

\noindent
  For $\nu \in \fsfn{\nats}{\calR}^\calL$ and $\sigma \in X^\nats$ we have, for $\ell \in \calL$ and $y \in Y$,
\begin{displaymath}
\begin{array}{rclcl}
  \multicolumn{3}{l}{ ( \, \fsfn{f}{\calR}^\calL \compose \eta_X \, )(\nu,\sigma)(\ell)(y) }
  \smallskip \\ & = & 
  \tssum_{x \myin \finv(y)} \; \eta_X( \nu, \sigma )( \ell )( x )
  \displaybreak[2]
  & = & 
  \tssum_{x \myin \finv(y)} \; \tssum_{n \myin \sigmainv(x)} \; \nu \mkern1mu ( \ell )( n )
  \displaybreak[2]
  \smallskip \\ & = & 
  \tssum_{n \myin ( f \compose \sigma)^{\mkern-1mu -1}} \; \nu \mkern1mu ( \ell )( n )
  & = & 
  \eta_Y \, (\nu, f \compose \sigma)( \myell )(y)
  \displaybreak[2]
  \smallskip \\ & = & 
  \eta_Y \, \bigl( \, ( \mathit{id}_{\fsfn{\nats}{\calR}^\calL} \times f^\nats \, ) (\nu, \sigma ) \bigr) (\myell)(y)
  & = & 
  ( \, \eta_Y \compose \: ( \, \mathit{id}_{\fsfn{\nats}{\calR}^\calL} \times f^\nats \, ))(\nu, \sigma )(\myell)(y)
\end{array}
\end{displaymath}
Thus, $\fsfn{f}{\calR}^\calL \compose \eta_X = \eta_Y \compose \: ( \, \mathit{id}_{\fsfn{\nats}{\calR}^\calL} \times f^\nats \,) $ and $\eta : \mathit{id}_{\fsfn{\nats}{\calR}^\calL} \times ( {\cdot} )^\nats \to \calVLR$ is a natural transformation.

  Finally, we check that $\eta_X : \mathit{id}_{\fsfn{\nats}{\calR}^\calL} \times X^\nats \to \calVLR(X)$ is surjective.
  Choose a set~$X$ and a mapping $\varphi : \calL \to \fsfn{X}{\calR}$. 
  Say, $\spt( \varphi(\myell) ) = \SET{ \, x_{\mkern1mu 1}^{\mkell} , \ldots , x_{n(\myell)}^\mkell }$. 
  Without loss of generality we assume $X \neq \spt( \varphi(\myell) )$ and pick $x^\mkell_0 \in X \backslash \spt( \varphi(\myell) )$.
  Define $\nu \in \fsfn{\nats}{\calR}^\calL$ by $\nu \mkern1mu ( \myell )(n) = \varphi( \myell )(x^\mkell_n)$ for $n = 1 \ldots n(\myell)$ and $\nu( \myell )(n) = 0$ otherwise.
  Define $\sigma : \nats \to X$ by $\sigma(n) = x_n^\mkell$ for $1 \leqslant n \leqslant n(\myell)$ and $\sigma(n) = x_0^\mkell$ otherwise. 
  Then we have 
\begin{displaymath}
\begin{array}{rclclcl}
  \eta_X \la \nu , \sigma \ra ( \myell )( x_i^\mkell )
  & = &
  \tssum_{m \myin \sigmainv(x_i^\mkell)} \; \nu \mkern1mu ( \myell )(m) 
  & = &
  \nu \mkern1mu ( \myell )(i) 
  & = & 
  \varphi \mkern1mu ( \myell )( x^\mkell_i ) 
  \\
  \eta_X \la \nu , \sigma \ra ( \myell )( x )
  & = &
  \tssum_{m \myin \sigmainv(x)} \  \ \nu \mkern1mu ( \myell )(m) 
  & = &
  \tssum_{n \myin \nats \backslash \SET{ \, 1, \ldots, n(\myell) \, }} \; \nu \mkern1mu (\myell)(n)
  & = &
  0
\end{array}
\end{displaymath}
for $i = 1 \ldots n(\myell)$ and $x \notin \spt( \nu ( \ell ) )$.
Thus $\eta_X \la \nu , \sigma \ra ( \myell )( x ) = \varphi \mkern1mu ( \myell )(x)$ for all~$\ell \in \calL$ and~$x \in X$, $\eta_X \la \nu , \sigma \ra  = \varphi$ and $\eta_X : \mathit{id}_{\fsfn{\nats}{\calR}^\calL} \times X^\nats \to \calVLR(X)$ is surjective.
\end{proof}

\subsection*{Additional proofs for Section~\ref{sec-pepa}}
  
\myblock{Definition}\cite[Definition 3.3.1]{Hil96:phd}
We put
\begin{displaymath}
\begin{array}{lcllclcll}
  r_a(\nil) & = & 0 &
  &&
  r_a( P \cho Q ) & = & r_a(P) + r_a(Q) &
  \\
  r_a( \alambda.P ) & = & \lambda &
  &&
  r_a( P \prlA Q ) & = & r_a(P) + r_a(Q)
  & \text{if $a \notin A$} \\ 
  r_a( \blambda.P ) & = & 0 & \text{for $b \neq a$}
  &&
  r_a( P \prlA Q ) & = & \min \SET{ \, r_a(P) ,\, r_a(Q) \, } 
  & \text{if $a \in A$}
  \\
  r_a(X) & = & r_a(P) & \text{if $X \dfas P$}
\end{array}
\end{displaymath}

\blankline

\myblock{Lemma~\ref{lm-sem-syn-arf}}
  Let $P \in \prcPEPA$ and $a \in \calA$. 
  Suppose $P \mtrans{\Edelaya}_p \amsetP$.
  Then $\fsum \amsetP = r_a(P)$.
\begin{proof}
  Guarded recursion.
  We treat the two cases for the parallel construct.
  
  Case $P = P_1 \prlA P_2$, $a \notin A$.
  Suppose $P_1 \mtrans{\Edelaya}_p \amsetP_1$, $P_2 \mtrans{\Edelaya}_p \amsetP_2$.
  Then $\amsetP = ( \, \amsetP_1 \prlA \chut_{P_2} \, ) \cho ( \, \chut_{P_1} \prlA \amsetP_2 \, )$.
  Therefore we have 
\begin{displaymath}
\def\arraystretch{1.2}
\begin{array}{rcll}
  \fsum \amsetP
  & = &
  \fsum ( \, \amsetP_1 \prlA \chut_{P_2} \, ) \cho ( \, \chut_{P_1} \prlA \amsetP_2 \, )
  \\ & = & 
  \fsum ( \, \amsetP_1 \prlA \chut_{P_2} \, ) \cho \fsum ( \, \chut_{P_1} \prlA \amsetP_2 \, )
  & \text{(by Lemma~\ref{lm-props-fsum})}
  \\ & = & 
  ( \, \fsum \amsetP_1 \cdot \fsum \chut_{P_2} \, ) \cho ( \, \fsum \chut_{P_1} \cdot \fsum \amsetP_2 \, )
  & \text{(by Lemma~\ref{lm-props-fsum})}
  \\ & = & 
  \fsum \amsetP_1 \cho \fsum \amsetP_2
  & \text{(since $\fsum \chut_{P_1}, \fsum \chut_{P_2} = 1$)}
  \\ & = & 
  r_a(P_1) \cho r_a(P_2)
  & \text{(by the induction hypothesis)}
  \\ & = & 
  r_a(P_1 \prlA P_2)
  & \text{(by definition~$r_a$)}
\end{array}
\def\arraystretch{1.0}
\end{displaymath}

  Case $P = P_1 \prlA P_2$, $a \in A$.
  Suppose $P_1 \mtrans{\Edelaya}_p \amsetP_1$, $P_2 \mtrans{\Edelaya}_p \amsetP_2$.
  Then $\amsetP = \arf ( \mkern1mu \amsetP_1, \amsetP_2 \mkern1mu ) \cdot ( \, \amsetP_1 \prlA  \amsetP_2 \, )$.
  If $\fsum \amsetP_1 , \fsum \amsetP_2 > 0$ we have
\begin{displaymath}
\def\arraystretch{1.2}
\begin{array}{rcll}
  \fsum \amsetP
  & = &
  \fsum ( \arf( \mkern1mu \amsetP_1, \amsetP_2 \mkern1mu ) \cdot ( \, \amsetP_1 \prlA \amsetP_2 \, )
  \\ & = & 
  \arf( \mkern1mu \amsetP_1, \amsetP_2 \mkern1mu ) \cdot \fsum \amsetP_1 \cdot  \fsum \amsetP_2 
  & \text{(by Lemma~\ref{lm-props-fsum})}
  \\ & = & 
  \begin{array}{c}
    \min \mkern2mu \SET{ \, \fsum \amsetP_1 ,\, \fsum \amsetP_2 \, }
    	\\ \hline
	\fsum \amsetP_1 \cdot \fsum \amsetP_2 )
  \end{array}
  \cdot \fsum \amsetP_1 \cdot \fsum \amsetP_2
  & \text{(by definition of $\arf$)}
  \\ & = & 
  \min \mkern2mu \SET{ \, \fsum \amsetP_1 , \, \fsum \amsetP_2 \, }
  & 
  \\ & = & 
  \min \mkern2mu \SET{ \, r_a( P_1 ), \, r_a( P_2 ) \, }
  & \text{(by the induction hypothesis)}
  \\ & = & 
  r_a(P_1 \prlA P_2)
  & \text{(by definition~$r_a$)}
\end{array}
\def\arraystretch{1.0}
\end{displaymath}
  If $\fsum \amsetP_1 , \fsum \amsetP_2 = 0$, then $\arf( \amsetP_1, \amsetP_2) = 0$, by definition, and $r_a(P_1), r_a(P_2) = 0$, by induction hypothesis.
  Therefore we have $\fsum \amsetP = \arf( \amsetP_1, \amsetP_2) \cdot \fsum \amsetP_1 \cdot \fsum \amsetP_2 = 0$ as well as $r_a(P_1 \prlA P_2) = \min \mkern2mu \SET{ \, r_a( P_1 ), \, r_a( P_2 ) \, } = 0$.
  So, also now, $\fsum \amsetP = r_a(P_1 \prlA P_2)$.
  The other cases are straightforward, in the case of $P_1 \cho P_2$  also relying on Lemma~\ref{lm-props-fsum}.
\end{proof}

\blankline

\myblock{Corollary}
  If $P \mtrans{\Edelaya} \amsetP$ and $Q \mtrans{\Edelaya} \amsetQ$, then $\arf( \amsetP , \amsetQ ) = \arf( P , Q )$.
\begin{proof}
  Direct from the definitions.
\end{proof}

\blankline

\myblock{Lemma~\ref{lm-cnt-ltfs-match}}
  Let $P \in \prcPEPA$ and $a \in \calA$. 
  Suppose $P \mtrans{\Edelaya} \amsetP$.
  Then it holds that $\amsetP (P') = \tssum \; \MSET{ \lambda }{ P \trans{a,\lambda} P' }$ for all~$P' \in \prcPEPA$.
\begin{proof}
  Guarded induction on~$P$.
  We only treat the cases for the parallel composition.
  Note, the operation $ {\prlA} : \prcPEPA \times \prcPEPA \to \prcPEPA$ with ${\prlA}\twotuple{P_1}{P_2} = P_1 \prlA P_2$ is injective.
  Recall, for $\amsetP_1, \amsetP_2 \in \fsfn{\prcPEPA}{\nnreals}$, we have $(\amsetP_1 \prlA \amsetP_2)(P_1 \prlA P_2) = \amsetP_1(P_1) \cdot \amsetP_2(P_2)$.
  
  Suppose $a \notin \calA$.
  Assume $P_1 \mtrans{a} \amsetP_1$, $P_2 \mtrans{a} \amsetP_2$, $P _1 \prlA P_2 \mtrans{a} \amsetP$.
  We distinguish three cases.
  Case~(I), $P' = P'_1 \prlA P_2$, $P'_1 \neq P_1$.
  Then we have
\begin{displaymath}
\def\arraystretch{1.2}
\begin{array}{rcll}
  \multicolumn{4}{l}{\tssum \MSET{ \lambda }{ P_1 \prlA P_2 \, \transalambda \, P' }}
  \\ & = &
  \tssum \MSET{ \lambda }{ P_1 \, \transalambda \, P'_1 }
  & \text{(by rule (PAR1a))}
  \\ & = &
  \amsetP_1( P'_1 )
  & \text{(by the induction hypothesis)}
  \\ & = &
  \amsetP_1( P'_1 ) \cdot \chut_{P_2}( P_2 )
  & \text{(as $\chut_{P_2}( P_2 ) = 1$)}
  \\ & = & \multicolumn{2}{l}{%
  ( \amsetP_1 \prlA \chut_{P_2} )( P'_1 \prlA P_2 )
  +
  ( \chut_{P_1} \prlA \amsetP_2 )( P'_1 \prlA P_2 )
  } 
  \\ & & &
  \text{(definition $\prlA$ on~$\fsfn{\prcPEPA}{\nnreals}$, $\chut_{P_1}(P'_1) = 0$)}
  \\ & = & 
  \amsetP( P' )
  & \text{(by rule (PAR1)}
\end{array}
\def\arraystretch{1.0}
\end{displaymath}
Case~(II), $P' = P_1 \prlA P'_2$, $P'_2 \neq P_2$: similar.
  Case~(III), $P' = P_1 \prlA P_2$.
  Then we have
\begin{displaymath}
\def\arraystretch{1.2}
\begin{array}{rcll}
  \multicolumn{4}{l}{\tssum \MSET{ \lambda }{ P_1 \prlA P_2 \, \transalambda \, P' }}
  \\ & = &
  \bigl( \, \tssum \MSET{ \lambda }{ P_1 \, \transalambda \, P_1 } \, \bigr)
  \  + \ 
  \bigl( \, \tssum \MSET{ \lambda }{ P_2 \, \transalambda \, P_2 } \, \bigr)
  & \text{(by rules (PAR1a) and (PAR1b))}
  \\ & = &
  \amsetP_1( P_1 )
  +
  \amsetP_2( P_2 )
  & \text{(by the induction hypothesis)}
  \\ & = & \multicolumn{2}{l}{%
  ( \amsetP_1 \prlA \chut_{P_2} )( P_1 \prlA P_2 )
  +
  ( \chut_{P_1} \prlA \amsetP_2 )( P_1 \prlA P_2 )
  } 
  \\ & & \multicolumn{2}{r}{%
  \text{(definition $\prlA$ on~$\fsfn{\prcPEPA}{\nnreals}$, $\chut_{P_1}(P_1),\ \chut_{P_2}( P_2 ) = 1$)}
  } 
  \\ & = & 
  \amsetP( P' )
  & \text{(again by rule (PAR1)}
\end{array}
\def\arraystretch{1.0}
\end{displaymath}

  Suppose $a \in A$.
  Assume $P_1 \mtrans{a} \amsetP_1$, $P_2 \mtrans{a} \amsetP_2$, $P _1 \prlA P_2 \mtrans{a} \amsetP$.  
  Without loss of generality, $P' = P'_1 \prlA P'_2$ for suitable $P'_1, P'_2 \in \prcPEPA$.
  
\begin{displaymath}
\def\arraystretch{1.2}
\begin{array}{rcll}
  \multicolumn{4}{l}{\tssum \MSET{ \lambda }{ P_1 \prlA P_2 \, \transalambda \, P' }}
  \displaybreak[4] \\ & = &
  \tssum \MSET{ \arf(P_1,P_2) \cdot \lambda_1 \cdot \lambda_2 }{ P_1 \, \trans{a,\lambda_1} \, P'_1 ,\, P_2 \, \trans{a,\lambda_2} \, P'_2 }
  & \text{(by rule (PAR2))}
  \displaybreak[4] \\ & = &
  \arf(P_1,P_2) \cdot 
  \Bigl( \, 
  \tssum \MSET{ \lambda_1 }{ P_1 \, \trans{a,\lambda_1} \, P'_1 } 
  \, \Bigr) \cdot \Bigl( \, 
  \tssum \MSET{ \lambda_2 }{ P_2 \, \trans{a,\lambda_2} \, P'_2 }
  \, \Bigr)
  & \text{(by distributivity)}
  \displaybreak[4] \\ & = &
  \arf(P_1,P_2) \cdot 
  \amsetP_1( P'_1 )
  \cdot
  \amsetP_2( P'_2 )
  & \text{(by the induction hypothesis)}
  \displaybreak[3] \\ & = &
  \arf( \amsetP_1, \amsetP_2) \cdot 
  \amsetP_1( P'_1 )
  \cdot
  \amsetP_2( P'_2 )
  & \text{(by the corollary above)}
  \displaybreak[2] \\ & = & 
  \arf(\amsetP_1,\amsetP_2) \cdot 
  ( \amsetP_1 \prlA \amsetP_2 )( P'_1 \prlA P'_2 )
  & \text{(definition $\prlA$ on~$\fsfn{\prcPEPA}{\nnreals}$)}
  \\ & = & 
  \amsetP( P' )
  & \text{(by rule (PAR2)}
\end{array}
\def\arraystretch{1.0}
\end{displaymath}
The other cases are simpler and omitted here.
\end{proof}

\subsection*{Additional proofs for Section~\ref{sec-iml}}

\myblock{Lemma~\ref{lm-mtrans-trans}}
\begin{itemize}
\item [(a)]
  Let $P \in \prcIML$ and $a \in \calA$. 
  If $P \mtrans{a}_1 \amsetP$ then $P \trans{a} P' \iff \amsetP( P' ) = \TRUE$.
\item [(b)]
  Let $P \in \prcIML$ . 
  If $P \mtrans{\Edelay}_2 \amsetP$ then $\tssum \MSET{ \lambda }{ P \Mtrans{\lambda} P' } = \amsetP( P' )$.
  \qed
\end{itemize}
\begin{proof} (a)
  Guarded induction. 
  Let $a \in \calA$.
  We treat the typical cases $\lambda.P$ and $P_1 \prlA P_2$ for $a \notin A$.
  
  Case~$\lambda.P$.
  Suppose $\lambda.P \mtrans{a} \amsetP$.
  Then we have $\amsetP = \zerof_\bools$.
  Thus, both $\lambda.P \trans{a} P'$ for no~$P' \in \prcIML$, as no transition is provided in~$\trans{}$, and $\amsetP( P' ) = \FALSE$ by definition of~$\zerof_\bools$, for all~$P' \in \prcIML$.
  
  Case~$P_1 \prlA P_2$, $a \notin A$.
  Suppose $P_1 \mtrans{a} \amsetP_1$, $P_2 \mtrans{a} \amsetP_2$ and $P_1 \prlA P_2 \mtrans{a} \amsetP$.
  Then it holds that $\amsetP = ( \amsetP_1 \prlA \chut_{P_2} ) \cho ( \chut_{P_1} \prlA \amsetP_2 )$.
  Recall, for~$Q \in \prcIML$ and $\chut_Q \in \fsfn{\prcIML}{\bools}$, $\chut_Q(Q') = \TRUE$ iff $Q' = Q$, for~$Q' \in \prcIML$.
  We have
\begin{displaymath}
\def\arraystretch{1.2}
\begin{array}{rcl}
  \multicolumn{3}{l}{%
  P_1 \prlA P_2 \trans{a} P'
  } 
  \\ & \Leftrightarrow &
  ( \, P_1 \trans{a} P'_1 \land P' = P'_1 \prlA P_2 \, ) 
  \; \lor \;
  ( \, P_2 \trans{a} P'_2 \land P' = P_1 \prlA P'_2 \, ) 
  \\ &&
  \text{(by analysis of $\trans{}$)}
  \\ & \Leftrightarrow &
  ( \, \amsetP_1( P'_1 ) = \TRUE \land P' = P'_1 \prlA P_2 \, )
  \; \lor \;
  ( \, \amsetP_2( P'_2 ) = \TRUE \land P' = P_1 \prlA P'_2 \, )
  \\ && \text{(by the induction hypothesis)}
  \displaybreak[2] \\ & \Leftrightarrow &
  ( \, \amsetP_1( P'_1 ) \cdot \chut_{P_2}(P_2) = \TRUE \land P' = P'_1 \prlA P_2 \, )
  \; \lor \;
  ( \, \chut_{P_1}(P_1) \cdot \amsetP_2( P'_2 ) = \TRUE \land P' = P_1 \prlA P'_2 \, )
  \\ && \text{(by definition of $\chut_{P_1}$ and~$\chut_{P_2}$)}
  \displaybreak[2] \\ & \Leftrightarrow &
  ( \, ( \amsetP_1 \prlA \chut_{P_2} )(P'_1 \prlA P_2) = \TRUE \land P' = P'_1 \prlA P_2 \, )
  \; \\ && \ \ \ \ \ \ \ \ \ \ \lor \;
  ( \, ( \chut_{P_1} \prlA \amsetP_2 )( P_1 \prlA P'_2 ) = \TRUE \land P' = P_1 \prlA P'_2 \, )
  \\ && \text{(by definition of~$\prlA$)}
  \displaybreak[2] \\ & \Leftrightarrow &
  ( \amsetP_1 \prlA \chut_{P_2} )(P') = \TRUE 
  \; \lor \;
  ( \chut_{P_1} \prlA \amsetP_2 )(P') = \TRUE 
  \\ && \text{(by definition of~$\prlA$, $\chut_{P_1}$ and~$\chut_{P_2}$)}
  \displaybreak[2] \\ & \Leftrightarrow &
  ( \, ( \amsetP_1 \prlA \chut_{P_2} ) \cho ( \chut_{P_1} \prlA \amsetP_2 ) \, )( P' ) = \TRUE
  \\ && \text{(by definition of~$\cho$ on~$\fsfn{\prcIML}{\bools}$)}
  \\ & \Leftrightarrow &
  \amsetP( P' ) = \TRUE
\end{array}
\def\arraystretch{1.0}
\end{displaymath}
  The other cases are standard or similar and easier.
%
%

  (b)
  Guarded induction.
  We treat the cases for~$\mu.P$ and $P_1 \prlA P_2$.  
  Case~$\mu.P$.
  Assume $P \mtrans{\Edelay}_2 \amsetP$.
  Suppose $P = \mu.P'$.
  Then it holds that $P$ admits a single $\Mtrans{} \mkern3mu$-transition, viz.\ $P \Mtrans{\mu} P'$.
  Thus we have $\tssum \MSET{ \lambda }{ P \Mtrans{\lambda} P' } = \mu = [ \, P' \mapsto \mu \, ](P') = \amsetP( P' )$.
  Suppose $P = \mu.P''$ for some $P'' \neq P$.
  Then we have $\tssum \MSET{ \lambda }{ P \Mtrans{\lambda} P' } = 0 = [ \, P'' \mapsto \mu \, ](P') = \amsetP( P' )$.
  
  Case~$P_1 \prlA P_2$.
  Assume $P_1 \mtrans{\Edelay}_2 \amsetP_1$, $P_2 \mtrans{\Edelay}_2 \amsetP_2$ and $P_1 \prlA P_2 \mtrans{\Edelay} \amsetP$. 
  It holds that $\amsetP = ( \amsetP_1 \prlA \chut_{P_2} ) \cho ( \chut_{P_1} \prlA \amsetP_2 )$.
  We calculate 
\begin{displaymath}
\def\arraystretch{1.2}
\begin{array}{rcl}
  \multicolumn{3}{l}{%
  \tssum \MSET{ \lambda }{ P_1 \prlA P_2 \Mtrans{\lambda} P' }
  } 
  \\ & = &
   \tssum \MSET{ \lambda }{ P_1 \Mtrans{\lambda} P'_1 ,\ P' = P'_1 \prlA P_2 } \cho \tssum \MSET{ \lambda }{ P_2 \Mtrans{\lambda} P'_2 ,\ P' = P_1 \prlA P'_2 }
  \\ && \text{(by analysis of $\Mtrans{}$)}
  \\ & = &
  ( \, \texttt{if} \; P' = P'_1 \prlA P_2 \; \texttt{then} \  \tssum \MSET{ \lambda }{ P_1 \Mtrans{\lambda} P'_1 } \ \texttt{else} \ 0 \ \texttt{end} \, ) 
  \cho 
  {} \\ & & \qquad \qquad
  ( \, \texttt{if} \; P' = P_1 \prlA P'_2 \; \texttt{then} \  \tssum \MSET{ \lambda }{ P_2 \Mtrans{\lambda} P'_2 } \  \texttt{else} \ 0 \ \texttt{end} \, )
  \displaybreak[2] \\ & = &
  ( \, \texttt{if} \; P' = P'_1 \prlA P_2 \; \texttt{then} \  \amsetP_1(P'_1) \ \texttt{else} \ 0 \ \texttt{end} \, ) 
  \cho 
  {} \\ & & \qquad \qquad
  ( \, \texttt{if} \; P' = P_1 \prlA P'_2 \; \texttt{then} \  \amsetP_2(P'_2) \  \texttt{else} \ 0 \ \texttt{end} \, )
  \\ && \text{(by induction hypothesis for $P_1$ and~$P_2$)}
  \displaybreak[2] \\ & = &
  ( \, \amsetP_1 \prlA \chut_{P_2} \, )(P') \cho ( \, \chut_{P_1} \prlA \amsetP_2 \, )(P')
  \\ && \text{(by definition of~$\prlA$, $\chut_{P_1}$,$\chut_{P_2}$ and~$\cho$ on~$\fsfn{\prcIML}{\nnreals}$)}
  \\ & = &
  \amsetP(P')
\end{array}
\def\arraystretch{1.0}
\end{displaymath}
  The remaining cases are left to the reader.
\end{proof}

\end{document}